\documentclass[11pt,letterpaper]{article}

\def\showauthornotes{0}

\def\showkeys{0}
\def\showdraftbox{0}
\def\showcolorlinks{1}
\def\usemicrotype{0}
\def\showfixme{0}

\usepackage{etex}

\usepackage[l2tabu, orthodox]{nag}

\usepackage{xspace,enumerate}

\usepackage[dvipsnames]{xcolor}

\usepackage[T1]{fontenc}
\usepackage[full]{textcomp}

\usepackage[american]{babel}

\usepackage{mathtools}

\usepackage{amsthm}

\newtheorem{theorem}{Theorem}[section]
\newtheorem*{theorem*}{Theorem}

\newtheorem{proposition}[theorem]{Proposition}
\newtheorem*{proposition*}{Proposition}
\newtheorem{lemma}[theorem]{Lemma}
\newtheorem*{lemma*}{Lemma}
\newtheorem{corollary}[theorem]{Corollary}

\newtheorem*{conjecture*}{Conjecture}

\newtheorem*{fact*}{Fact}

\newtheorem*{hypothesis*}{Hypothesis}

\theoremstyle{definition}
\newtheorem{definition}[theorem]{Definition}
\newtheorem*{definition*}{Definition}

\newtheorem{algorithm}[theorem]{Algorithm}

\newtheorem*{problem*}{Problem}

\theoremstyle{remark}
\newtheorem{claim}[theorem]{Claim}
\newtheorem*{claim*}{Claim}

\newtheorem*{remark*}{Remark}
\newtheorem{observation}[theorem]{Observation}
\newtheorem*{observation*}{Observation}

\usepackage[
letterpaper,
top=1in,
bottom=1in,
left=1in,
right=1in]{geometry}

\usepackage{amsmath,amsfonts,amssymb}

\usepackage{newpxtext} % T1, lining figures in math, osf in text
\usepackage{textcomp} % required for special glyphs
\usepackage[varg,bigdelims]{newpxmath}
\usepackage[scr=rsfso]{mathalfa}% \mathscr is fancier than \mathcal
\usepackage{bm} % load after all math to give access to bold math
\let\mathbb\varmathbb

\ifnum\showkeys=1
\usepackage[color]{showkeys}
\fi

\ifnum\showcolorlinks=1
\usepackage[
colorlinks=true,
urlcolor=blue,
linkcolor=blue,
citecolor=OliveGreen,
]{hyperref}
\fi

\ifnum\showcolorlinks=0
\usepackage[
pagebackref,
colorlinks=false,
pdfborder={0 0 0}
]{hyperref}
\fi

\usepackage{prettyref}
\newcommand{\pref}{\prettyref}

\newcommand{\savehyperref}[2]{\texorpdfstring{\hyperref[#1]{#2}}{#2}}

\newrefformat{eq}{\savehyperref{#1}{\textup{(\ref*{#1})}}}
\newrefformat{prog}{\savehyperref{#1}{\textup{Program \ref*{#1}}}}
\newrefformat{eqn}{\savehyperref{#1}{\textup{(\ref*{#1})}}}
\newrefformat{lem}{\savehyperref{#1}{Lemma~\ref*{#1}}}
\newrefformat{def}{\savehyperref{#1}{Definition~\ref*{#1}}}
\newrefformat{thm}{\savehyperref{#1}{Theorem~\ref*{#1}}}
\newrefformat{cor}{\savehyperref{#1}{Corollary~\ref*{#1}}}
\newrefformat{cha}{\savehyperref{#1}{Chapter~\ref*{#1}}}
\newrefformat{sec}{\savehyperref{#1}{Section~\ref*{#1}}}
\newrefformat{app}{\savehyperref{#1}{Appendix~\ref*{#1}}}
\newrefformat{tab}{\savehyperref{#1}{Table~\ref*{#1}}}
\newrefformat{fig}{\savehyperref{#1}{Figure~\ref*{#1}}}
\newrefformat{hyp}{\savehyperref{#1}{Hypothesis~\ref*{#1}}}
\newrefformat{alg}{\savehyperref{#1}{Algorithm~\ref*{#1}}}
\newrefformat{rem}{\savehyperref{#1}{Remark~\ref*{#1}}}
\newrefformat{item}{\savehyperref{#1}{Item~\ref*{#1}}}
\newrefformat{step}{\savehyperref{#1}{step~\ref*{#1}}}
\newrefformat{conj}{\savehyperref{#1}{Conjecture~\ref*{#1}}}
\newrefformat{fact}{\savehyperref{#1}{Fact~\ref*{#1}}}
\newrefformat{prop}{\savehyperref{#1}{Proposition~\ref*{#1}}}
\newrefformat{prob}{\savehyperref{#1}{Problem~\ref*{#1}}}
\newrefformat{claim}{\savehyperref{#1}{Claim~\ref*{#1}}}
\newrefformat{relax}{\savehyperref{#1}{Relaxation~\ref*{#1}}}
\newrefformat{red}{\savehyperref{#1}{Reduction~\ref*{#1}}}
\newrefformat{part}{\savehyperref{#1}{Part~\ref*{#1}}}
\newrefformat{obs}{\savehyperref{#1}{Observation~\ref*{#1}}}
\newrefformat{corr}{\savehyperref{#1}{Corollary~\ref*{#1}}}

\newcommand{\Sref}[1]{\hyperref[#1]{\S\ref*{#1}}}

\usepackage{nicefrac}
\ifnum\usemicrotype=1
\usepackage{microtype}
\fi

\ifnum\showauthornotes=1
\newcommand{\Authornote}[2]{{\sffamily\small\color{red}{[#1: #2]}}}
\newcommand{\Authornotecolored}[3]{{\sffamily\small\color{#1}{[#2: #3]}}}
\newcommand{\Authorcomment}[2]{{\sffamily\small\color{gray}{[#1: #2]}}}
\newcommand{\Authorstartcomment}[1]{\sffamily\small\color{gray}[#1: }

\newcommand{\Authorfnote}[2]{\footnote{\color{red}{#1: #2}}}
\newcommand{\Authorfixme}[1]{\Authornote{#1}{\textbf{??}}}
\newcommand{\Authormarginmark}[1]{\marginpar{\textcolor{red}{\fbox{\Large #1:!}}}}
\else
\newcommand{\Authornote}[2]{}
\newcommand{\Authornotecolored}[3]{}
\newcommand{\Authorcomment}[2]{}
\newcommand{\Authorstartcomment}[1]{}

\newcommand{\Authorfnote}[2]{}
\newcommand{\Authorfixme}[1]{}
\newcommand{\Authormarginmark}[1]{}
\fi

\ifnum\showfixme=0

\fi

\usepackage{boxedminipage}

\newcommand{\Brac}[1]{\left[#1\right]}

\newcommand{\abs}[1]{\lvert#1\rvert}
\newcommand{\Abs}[1]{\left\lvert#1\right\rvert}

\newcommand{\norm}[1]{\lVert#1\rVert}

\newcommand{\Esymb}{\mathbb{E}}
\newcommand{\Psymb}{\mathbb{P}}
\newcommand{\Vsymb}{\mathbb{V}}

\DeclareMathOperator*{\E}{\Esymb}
\DeclareMathOperator*{\Var}{\Vsymb}
\DeclareMathOperator*{\ProbOp}{\Psymb}

\renewcommand{\Pr}{\ProbOp}

\newcommand{\Ex}[2][]{\E_{{#1}}\Brac{#2}}

\newcommand{\sos}{{SoS}\xspace }

\newcommand{\tensor}{\otimes}

\newcommand{\textparen}[1]{\text{(#1)}}

\ifx\because\undefined
\newcommand{\because}[1]{\textparen{because #1}}
\else
\renewcommand{\because}[1]{\textparen{because #1}}
\fi

\newcommand{\defeq}{\stackrel{\mathrm{def}}=}

\newcommand{\from}{\colon}

\newcommand\bdot\bullet

\ifx\mathds\undefined % use double stroke fonts if available
\DeclareMathOperator{\Ind}{\mathbb{I}}
\else
\DeclareMathOperator{\Ind}{\mathds 1}}
\fi

\DeclareMathOperator{\LP}{LP}

\DeclareMathOperator{\opt}{opt}

\DeclareMathOperator{\argmax}{argmax}

\DeclareMathOperator{\supp}{supp}

\newcommand{\etal}{et al.\xspace}

\newcommand{\N}{\mathbb N}
\newcommand{\R}{\mathbb R}
\newcommand{\F}{\mathbb F}

\newcommand{\problemmacro}[1]{\texorpdfstring{\textup{\textsc{#1}}}{#1}\xspace}

\newcommand{\maxksat}{\problemmacro{max $k$-sat}}

\newcommand{\cA}{\mathcal A}

\newcommand{\cF}{\mathcal F}

\newcommand{\cP}{\mathcal P}

\newcommand{\cS}{\mathcal S}

\usepackage{mathrsfs}

\ifnum\showdraftbox=1

\else

\fi

\let\epsilon=\varepsilon

\numberwithin{equation}{section}

\newcommand\MYcurrentlabel{xxx}

\newcommand{\MYstore}[2]{%
  \global\expandafter \def \csname MYMEMORY #1 \endcsname{#2}%
}

\newcommand{\MYload}[1]{%
  \csname MYMEMORY #1 \endcsname%
}

\newcommand{\MYnewlabel}[1]{%
  \renewcommand\MYcurrentlabel{#1}%
  \MYoldlabel{#1}%
}

\newcommand{\MYdummylabel}[1]{}

\newcommand{\torestate}[1]{%
  \let\MYoldlabel\label%
  \let\label\MYnewlabel%
  #1%
  \MYstore{\MYcurrentlabel}{#1}%
  \let\label\MYoldlabel%
}

\newcommand{\restatetheorem}[1]{%
  \let\MYoldlabel\label
  \let\label\MYdummylabel
  \begin{theorem*}[Restatement of \prettyref{#1}]
    \MYload{#1}
  \end{theorem*}
  \let\label\MYoldlabel
}

\newcommand{\restatedef}[1]{%
  \let\MYoldlabel\label
  \let\label\MYdummylabel
  \begin{definition*}[Restatement of \prettyref{#1}]
    \MYload{#1}
  \end{definition*}
  \let\label\MYoldlabel
}

\newcommand{\restatelemma}[1]{%
  \let\MYoldlabel\label
  \let\label\MYdummylabel
  \begin{lemma*}[Restatement of \prettyref{#1}]
    \MYload{#1}
  \end{lemma*}
  \let\label\MYoldlabel
}

\newcommand{\restateprop}[1]{%
  \let\MYoldlabel\label
  \let\label\MYdummylabel
  \begin{proposition*}[Restatement of \prettyref{#1}]
    \MYload{#1}
  \end{proposition*}
  \let\label\MYoldlabel
}

\newcommand{\restatefact}[1]{%
  \let\MYoldlabel\label
  \let\label\MYdummylabel
  \begin{fact*}[Restatement of \prettyref{#1}]
    \MYload{#1}
  \end{fact*}
  \let\label\MYoldlabel
}

\newcommand{\restateobs}[1]{%
  \let\MYoldlabel\label
  \let\label\MYdummylabel
  \begin{observation*}[Restatement of \prettyref{#1}]
    \MYload{#1}
  \end{observation*}
  \let\label\MYoldlabel
}
\newcommand{\restate}[1]{%
  \let\MYoldlabel\label
  \let\label\MYdummylabel
  \MYload{#1}
  \let\label\MYoldlabel
}

\newcommand{\eps}{\epsilon}

\let\origparagraph\paragraph
\renewcommand{\paragraph}[1]{\origparagraph{#1.}}

\allowdisplaybreaks

\usepackage[newenum,newitem,increaseonly]{paralist}
\usepackage{enumitem}

\usepackage{comment}

\let\citet\cite

\theoremstyle{definition}

\DeclareUrlCommand\email{}

\newcommand{\restateproblem}[2]{%
  \let\MYoldlabel\label
  \let\label\MYdummylabel
  \begin{problem*}[Restatement of \prettyref{#1}, {#2}]
    \MYload{#1}
    \end{problem*}
  \let\label\MYoldlabel
}
\usepackage{bm}

\usepackage{tikz,bbm}

\newcommand{\ovec}{{\vec 1}}

\usepackage{tikz}

\usepackage{relsize}
\usepackage[font=footnotesize]{caption}
\usepackage[flushleft,para]{threeparttable}
\makeatletter
\g@addto@macro\TPT@defaults{\footnotesize}
\makeatother

\newcommand{\Enk}{E_{n,k}}
\newcommand{\D}{\mathcal{D}}
\newcommand{\U}{\mathcal{U}}

\DeclareUrlCommand\email{}

\renewcommand{\it}{\em}

\DeclareMathOperator{\CSP}{CSP}

\DeclarePairedDelimiter\ceil{\lceil}{\rceil}

\newcommand{\maxkxor}{\problemmacro{max $k$-xor}}

\DeclareUrlCommand\email{}

\newcommand{\inner}[1]{\langle #1 \rangle}

\let\pref=\prettyref

\newcommand{\setcomp}[1]{\overline{#1}}
\begin{document}

% Declarations for Front Matter

\title{Extended Formulation Lower Bounds for Refuting Random CSPs}
\author{Jonah Brown-Cohen \thanks{Research supported by the Approximability and Proof Complexity project funded by the Knut and Alice Wallenberg Foundation.}\\
	KTH Royal Institute of Technology
\and Prasad Raghavendra \thanks{Supported by NSF CCF 1718695 and
	NSF CCF 1408673.}
\\
U.C.\ Berkeley}
\maketitle
\thispagestyle{empty}

\begin{abstract}
	Random constraint satisfaction problems (CSPs) such as random $3$-SAT are conjectured to be computationally intractable.
	The average case hardness of random $3$-SAT and other CSPs has broad and far-reaching implications on problems in approximation, learning theory and cryptography.

	In this work, we show subexponential lower bounds on the size of linear programming relaxations for refuting random instances of constraint satisfaction problems.
	Formally, suppose $P : \{0,1\}^k \to \{0,1\}$ is a predicate that supports a $t-1$-wise uniform distribution on its satisfying assignments.
	Consider the distribution of random instances of CSP $P$ with $m = \Delta n$ constraints.
	We show that any linear programming extended formulation that can refute instances from this distribution with constant probability must have size at least $\Omega\left(\exp\left(\left(\frac{n^{t-2}}{\Delta^2}\right)^{\frac{1-\nu}{k}}\right)\right)$ for all $\nu > 0$.
	For example, this yields a lower bound of size $\exp(n^{1/3})$ for random $3$-SAT with a linear number of clauses.

	We use the technique of {\it pseudocalibration} to directly obtain extended formulation lower bounds from the planted distribution.
	This approach bypasses the need to construct Sherali-Adams integrality gaps in proving general LP lower bounds.
	As a corollary, one obtains a self-contained proof of subexponential Sherali-Adams LP lower bounds for these problems.
	We believe the result sheds light on the technique of pseudocalibration, a promising but conjectural approach to LP/SDP lower bounds.

\end{abstract}

\clearpage

%\thispagestyle{empty}
%\tableofcontents
%\clearpage
\setcounter{page}{1}

\section{Introduction}

Constraint satisfaction problems (CSP) are a rich and expressive family of optimization problems that arise ubiquitously both in theory and practice.
There is a deep and extensive theory around the computational complexity of CSPs in the worst case.
%
%with CSPs such as $3$-SAT being prime examples of NP-complete problems.  
%
In addition to their worst-case intractability, CSPs are among the simplest examples of computational problems that appear to be intractable on average.
%
%Specifically, random instances of constraint satisfaction problems are seemingly intractable for a wide range of the parameters.
%
Arguably, random CSPs yield the most natural distributions of optimization problems that are both computationally intractable and easy to sample from.
Furthermore, the average case hardness of random CSPs such as {\it random $3$-SAT} and {\it random $k$-XOR} has been a useful starting point towards establishing intractability of various problems both in the average and the worst case \cite{AAMMW11,BKS13,DLS13,D16}.
This work is devoted to showing the computational intractability of random CSPs for a natural family of algorithms namely the class of linear programming relaxations also referred to as LP extended formulations.

\paragraph{Refuting random CSPs}

We first describe the main computational problem that we study: that of refuting random CSPs. 
Given a predicate $P : \{0,1\}^k \to \{0,1\}$, an instance of the corresponding CSP consists of a set of boolean variables $V$ and a family of constraints which consist of the predicate $P$ applied to subsets of variables and their negations.
For the sake of concreteness, we will use  $k$-SAT as our running example, where the predicate is $P(x_1,\ldots,x_k) = x_1 \vee x_2 \ldots  \vee x_k$ and the constraints are referred to as clauses.
An instance $\Phi$ of random $k$-SAT is sampled by drawing $m$ clauses independently over $n$ variables uniformly at random.  
The ratio $\Delta = m/n$ is termed the density of clauses in $\Phi$.
The focus of this work is the regime where the number of clauses $m \gg n$.
In this regime, a random $k$-SAT instance $\Phi$ is far-from-satisfiable with overwhelming probability.
More precisely, with overwhelming probability, no assignment satisfies more than a $1-\delta$ fraction of clauses for some absolute constant $\delta > 0$.
A natural computational task for random $k$-SAT in this regime is that of strong refutation, wherein the algorithm certifies an upper bound on the maximum number of clauses satisfied by any assignment.  Formally, a strong refutation algorithm is defined as follows. 
\begin{definition} (Strong Refutation)
An algorithm $\cA$ is a {\it strong refutation} algorithm for random $k$-SAT at density $\Delta$, if for a fixed constant $\eta > 0$, given a random instance $\Phi$ of $k$-SAT with density $\Delta$, the algorithm $\cA$:
\begin{compactitem}
\item Outputs YES with probability at least $\frac{1}{2}$ over the choice of $\Phi$.
\item Outputs NO if $\Phi$ has an assignment satisfying at least a $(1-\eta)$-fraction of clauses.
\end{compactitem}
\end{definition}
We will also use the terminology {\it $\eta$-strongly refutes} when we wish to specify the constant $\eta$ in the above definition.
Note that the task of strong refutation as described above is more stringent than the task of {\it weak refutation} or what is sometimes just called {\it refutation}.  In {\it weak refutation}, the algorithm outputs $NO$ if $\Phi$ has an assignment satisfying all the clauses, i.e., the algorithm is only required to certify that a given instance is not completely satisfiable.  Weak refutation is strictly easier than strong refutation in many cases such as random $k-XOR$. 

The task of refuting random CSPs arises in myriad areas of theoretical computer science including proof complexity \cite{BB02}, inapproximability \cite{Feige02}, SAT solvers, cryptography \cite{ABW10}, learning theory \cite{DLSS14}, statistical physics \cite{CLP02} and complexity theory \cite{BKS13}.

Using the spectral approach pioneered by Goerdt and Krivelevich \cite{GK01}, a series of works \cite{FG01,FGK05, COGL04,AOW15,BM15} have culminated in a polynomial-time strong refutation algorithm for random $k$-SAT at densities $\Delta \geq n^{\frac{k - 2}{2}}$.
Recall that a random $k$-SAT instance is far from being satisfiable with high probability for all $\Delta$ larger than some constant.
This leaves open a wide range of densities from constant to $n^{\frac{k - 2}{2}}$ where no polynomial time algorithms for refutation are known.
It was recently shown that degree $\tilde{O}(n^{\epsilon})$ sum-of-squares semidefinite programs can refute random $k$-SAT instances at density $n^{(\frac{k -2}{2})(1- \epsilon)}$ for all $\epsilon \in (0,1)$, thereby yielding a sub-exponential time refutation algorithm for all densities in this range \cite{RaghavendraRS16}.
However, the problem of refutation is widely believed to have super-polynomial computational complexity for this entire range of clause densities.

\paragraph{Linear and Semidefinite Programming Relaxations for CSPs}
With even the worst-case complexity of $\textsf{P} \neq \textsf{NP}$ still out of reach, the average case complexity of random CSPs is well beyond the realm of current techniques to conclusively settle.
Therefore, an approach to gather evidence towards hardness of refuting random CSPs is to consider restricted classes of algorithms.
Linear and semidefinite programming relaxations are the prime candidates to pursue, since they yield the best known algorithms for CSPs in many worst-case settings. 
\textbf{Specific} linear and semidefinite programming relaxations such as the Sherali-Adams LP hierarchy and the Sum-of-Squares SDP hierarchy have been extensively studied, and a fairly detailed and clear picture has emerged in the literature.
In \pref{sec:relatedwork} we briefly survey the lower bound results for the sum-of-squares SDP (also known as the Lasserre/Parrilo SDP hierarchy), which is the most powerful of the specific LP/SDP hierarchies (see \cite{kothari2017sum} for a detailed survey).
%

%Grigoriev \cite{Gri01} and Schoenebeck \cite{Sch08} rule out efficient weak refutations for random $k$-XOR and random $k$-SAT at densities significantly smaller than $m/n < n^{k/2-1}$.
%%
%Specifically, Schonebeck's result implies that with high probability over $k$-XOR instances $\Phi$ with clause density $m/n < O(n^{(k/2 -1)(1-\delta)})$, the \sos hierarchy cannot even weakly refute $\Phi$ at degree $O(n^{\delta})$.
%%
%More generally, Kothari \etal \cite{kothari2017sum} showed that for a predicate $P$ which supports a $t$-wise uniform probability distributions, the corresponding random CSP requires degree $\Omega(\frac{n}{\Delta^{2/(t-1)} \log \Delta})$ sum-of-squares proofs to weakly refute.  
%%
%In other words, the sum-of-squares SDP hierarchy takes super-polynomial time for all constraint densities less than $\tilde{\Omega}(n^{(t+1)/2})$, almost matching the upper bounds in \cite{RaghavendraRS16}.

\paragraph{Extended Formulations}
In a foundational work, Yannakakis \cite{Yannakakis91} laid out the model of LP extended formulations to capture the notion of a general linear program for a problem.
Roughly speaking, an {\it extended formulation} for a problem $\Pi$ consists of a family of linear programs for every input size $n$ with
the key restriction being that the feasible region of the linear program depends solely on the input size $n$, and is independent of the specific instance of the problem  
(analogous to having a single circuit independent of the instance for each input size).
We will defer the formal definition of extended formulations to later.

The model of extended formulations captures several classic linear programming based algorithms in combinatorial optimization such as that for minimum spanning tree.
%
%A notable exception is that of matching in general graphs which has no small extended formulations \cite{Rothvoss17}, but still admits polynomial time algorithms. 
%
In the case of CSPs, all of the LP/SDP hierarchies such as Sherali-Adams, Lovasz-Schriver, Sum-of-Squares SDP are special cases of LP/SDP extended formulations respectively. See \pref{sec:relatedwork} for a summary of the history of lower bounds for extended formulations.
%

%Yannakakis \cite{Yannakakis91,Yannakakis88} showed exponential lower bounds on the size of {\it symmetric} linear programs for the Travelling Salesman Problem (TSP) and Matching problems.
%%
%Three decades later, Fiorini \etal \cite{FioriniKaibelPashkovichTheis11,FMPTW12}, showed how to avoid the symmetry assumption and show a lower bound for general assymmetric linear programming relaxations of the TSP polytope.  
%%
%This was followed by a flurry of recent work on lower bounds for extended formulations for various problems.
%%
%Rothvo{\ss} showed an exponential lower bound on the size of extended formulations for the Matching problem \cite{rothvoss2013matching}.
%%

In the context of \textbf{worst-case} CSPs, the work of \cite{chan2013approximate,KothariMR17} yields a general "lifting" theorem which takes lower bounds for Sherali-Adams linear programs to corresponding lower bounds for LP extended formulations.
On the one hand, this establishes optimality of Sherali-Adams linear programs among all linear programs of comparable size for CSPs.
Furthermore, with the wealth of lower bound results known for Sherali-Adams linear programs, one immediately obtains lower bounds for extended formulations.
%
%An analogous result to lift lower bounds on Sum-of-Squares SDPs to lower bounds for SDP extended formulations was shown in \cite{DBLP:conf/stoc/LeeRS15}.

\origparagraph{Average-case lower bounds for general LPs?}
The previous work mentioned above yields average-case lower-bounds for \textbf{specific} LP relaxations, and \textbf{worst-case} lower bounds for general LP extended formulations. Thus, a natural question is whether we can simultaneously obtain average-case extended formulation lower bounds.
We answer this question in the affirmative, and as an added bonus our extended formulation lower bounds are significantly stronger than those known before, even for worst-case instances.

\subsection{Our Results}
Our main result is a lower bound on LP extended formulations for strongly refuting random CSPs.  
\begin{theorem}
	\label{thm:lplowerbound}
Let $P$ be a predicate supporting a $(t-1)$-wise uniform distribution on satisfying assignments.
For any $\eps > 0$, let $1 < \Delta < n^{\frac{t-2}{2} - \eps}$.	
Then for any constants $\eta, \nu > 0$, any linear programming formulation that $\eta$-strongly refutes random instances of $\CSP(P)$ with constant probability must have size at least $\Omega\left(\exp\left(\left(\frac{n^{t-2}}{\Delta^2}\right)^{\frac{(1 - \nu)}{k}}\right)\right)$.
\end{theorem}
By using the fact that the $k$-XOR predicate and the $k$-SAT predicate support $(k-1)$-wise uniform distribution on satisfying assignments, the following corollary is immediate.
\begin{corollary}
For any $\eps > 0$, let $1 < \Delta < n^{\frac{k-2}{2} - \eps}$.
Then for any constants $\eta, \nu > 0$, any linear programming formulation that $\eta$-strongly refutes random instances of $\CSP(P)$ with constant probability must have size at least $\Omega\left(\exp\left(\left(\frac{n^{k-2}}{\Delta^2}\right)^{\frac{(1 - \nu)}{k}}\right)\right)$.
\end{corollary}
The above lower bounds on the size of linear programming extended formulations are quantitatively stronger than all known extension complexity results for these problems, even for worst-case instances.
For example, in the case of $k$-XOR or $k$-SAT with linearly many clauses, the lower bound of Kothari \etal \cite{KothariMR17} is $\exp(n^{\epsilon_0})$ for some unspecified small constant $\epsilon_0$, while the above result yields a lower bound of $\exp(n^{1- \frac{2}{k} - \epsilon})$ for all $\epsilon > 0$.

For specific LP/SDP hierarchies such as the Sherali-Adams LP or the sum-of-squares SDP, lower bounds in literature are both quantitatively stronger (see \cite{benabbas2012sdp,Gri01,Sch08,kothari2017sum}) and rule out weak refutations too. 
However, it is interesting that our proof recovers sub-exponential lower bounds for Sherali-Adams linear programs to strongly refute these CSPs, as a corollary of the more general LP extension complexity lower bound.

\paragraph{Extended formulation lower bounds for the uniform distribution}
Our results improve on prior extended formulation lower bounds for CSPs in two ways. 
First, ours are the first lower bounds that apply to the natural uniform distribution on CSPs.
Second, as mentioned above, our lower bounds are significantly quantitatively stronger than all prior results, including in the worst case.
%In order to achieve this first improvement, we are unable to use the approach from prior work of lifting Sherali-Adams LP lower bounds to general extended formulation lower bounds. Interestingly, finding a way around this approach is also what allows us to obtain the quantitative improvement mentioned above.
We now give a more detailed discussion of the relation of our results to those implied by the extended formulation lower bounds in \cite{chan2013approximate, KothariMR17, kothari2017sum, DBLP:conf/stoc/LeeRS15}.

Prior work does in fact give lower bounds for \textit{some} distribution over CSP instances, but this distribution is far from uniform.
To see why, recall that the lifting theorems in \cite{chan2013approximate, KothariMR17} allow one to take lower bounds for Sherali-Adams LPs to corresponding lower bounds for general LPs.
Specifically, fix an instance $I$ of a CSP such that the $d$-round Sherali-Adams LP optimum is $c$, but the integral optimum is $s$ for some $c,s \in \R$.
In other words, $I$ is a $c$ vs $s$-integrality gap instance for the $d$-round Sherali-Adams LP.
The result of \cite{chan2013approximate} implies a lower bound of $n^{\omega(d)}$ on linear programs that yield a $c$ vs $s$ approximation on a family of instances $\cF_I$ constructed as follows.
Each instance $I' \in \cF_I$ is an instance on $n$ variables $V$ where the instance $I$ is planted on some subset $S$ of $m$ variables, while there are no constraints on the remaining $n-m$ variables.
%
%For computational purposes, each instance $I' \in \cF_I$ consists of essentially the original gap instance $I$ planted on a random subset of variables.

By \cite{kothari2017sum}, the integrality gap instance $I$ can be chosen to be a random CSP instance of size $m$.  This yields an extended formulation lower bound for the distribution where first a random CSP instance $I$ on $m$ variables is chosen, and then the instance $I$ is planted in a larger $n$ variable instance, with the remaining $n-m$ dummy variables not appearing in any constraints.
Although this distribution on sparse instances is entirely different from the natural random CSP distribution, it seems that
in principle, the addition of the dummy variables should not change the computational complexity of the underlying problem.

However, there cannot be a generic black-box reduction to translate lower bounds for the family $\cF_I$, into a hardness result for the underlying instance with the dummy variables removed.
The reason for this is that LP extended formulations are a non-uniform model of computation, where we allow a different LP feasibility polytope for each instance size.
In particular, if we start with a single, fixed integrality gap instance $I$ on $m$ variables, then the family $\cF_I$ admits no small LP relaxations.
But the original instance $I$--being a single instance of the problem on $m$ variables--always admits a small LP relaxation, since the LP feasibility polytope at each size $m$ can just be tailored to this specific instance.
This obstacle to proving LP extended formulation lower bounds is analogous to the difference between a circuit lower bound proved under a sparse distribution coming from random restrictions versus the uniform distribution over inputs.

Given the challenges outlined above, we now explain what is required in order for us to prove lower bounds for the uniform distribution.
At a high-level, the proof of the lifting theorems in \cite{chan2013approximate} can be thought of as taking a random restriction on the instance, and using a Sherali-Adams LP solution on this restriction.
%However, random restrictions cannot be directly used towards proving a lower bound under the natural uniform distribution over CSP instances.
%%
%Thus, the distinction is analogous to that between a circuit lower bound under a sparse distribution vs the uniform distribution over inputs. 
%%
%This lack of random restriction poses some technical challenges circumventing which lead to consequences that are interesting in their own right.
%%
By virtue of the random restriction, \cite{chan2013approximate} can use the Sherali-Adams LP solution for random instances constructed in \cite{kothari2017sum} as a blackbox.
Since we are targeting the natural distribution on CSP instances, we cannot afford to take a random restriction.
Thus, the Sherali-Adams solution can no longer be treated as a blackbox, but one needs to appeal to the properties of the function 
$$\Lambda : \{\text{ random instances } I\} \to \{\text{LP solution for I}\}$$
While the prior work \cite{kothari2017sum}, constructs the solution $\Lambda(I)$ for a given instance $I$, it is difficult to extract the properties of the function $\Lambda$ from these constructions.
In light of this, it is necessary to use a different approach to systematically construct the function $\Lambda$, so that we can directly reason about its properties.
To this end, we appeal to the technique of pseudocalibration developed in the work on SDP lower bounds for planted clique \cite{barak2016nearly}.

\paragraph{Pseudocalibration}
In this paper, we use pseudocalibration to explicitly construct the candidate Sherali-Adams LP solution $\Lambda(I)$ for an instance $I$ as a low-degree polynomial. 
This is the key to our analysis, as it allows us to directly reason about the properties of $\Lambda(I)$ in a non-blackbox manner.
The basic idea of pseudocalibration is to find a planted distribution $\D_*$ that locally looks the same
as the uniform distribution $\D$ on CSPs, but which is only supported on instances with a very
different objective function value than those coming from $\D$. A good example of such a planted
distribution is simple to describe for the case of random $k$-SAT: first sample a uniform random
satisfying assignment $x$, and then sample the constraints in the instance $I$ so that $x$ is a satisfying
assignment to all the constraints.

Next we write the density of $\D_*$ as a polynomial in the instance $I$. We then define $\Lambda(I)$ to be
the low-degree projection of the density of $\D_*$, and let this be our candidate Sherali-Adams LP
solution. Since the Sherali-Adams LP in some sense only looks at local properties of the instance,
and our two distributions are locally similar, we might hope that this candidate solution satisfies
all the Sherali-Adams LP constraints. Simultaneously, we hope that the Sherali-Adams objective
function value remains the same as the original objective function value under $\D_*$, giving rise to a
good integrality gap instance.

For our purposes, the power of this technique comes from the fact that we are able to directly
analyze the function $\Lambda(I)$ constructed in this way. In particular, we compute explicit bounds on
the coefficients of the polynomial $\Lambda(I)$, and use these to directly prove LP extended formulation
lower bounds. Note that LP extension complexity lower bounds subsume Sherali-Adams LP
lower bounds. Therefore, we recover a sub-exponential lower bound against Sherali-Adams LP
relaxations on random instances without explicitly constructing a feasible solution for any random
instance $I$. Moreover, our lower bound proof only relies on local statistics of the planted distribution
$\D_*$ without ever appealing to a global property such as expansion of the constraint graph. We
believe that this application of pseudocalibration sheds light on the technique, and that the ideas
could potentially be useful towards formally implementing pseudocalibration for a broader family
of problems.

\subsection{Related Work}
\label{sec:relatedwork}
There are many applications of the computational complexity of random $k$-SAT to hardness of approximation, learning theory and cryptography.
In particular, for $k = 3$ and constant clause densities, Feige's ``R3SAT hypothesis,'' asserts that for any $\delta>0$, there exists some constant $c$ such that there is no polynomial-time algorithm certifying that a random $3$-SAT instance has value at most $1-\delta$ (that is, strongly refutes $3$-SAT) at clause density $m/n = c$.
Feige used this hypothesis to obtain simple proofs of hardness of approximation results for problems such as densest-$k$ subgraph and min-bisection \cite{Feige02}, which were otherwise mostly out of reach for standard PCP based reductions.  
In turn, these problems have been used to derive several other hardness of approximation results \cite{demaine2008combination,briest2008uniform,alon2012optimizing}.
%
%This hypothesis has subsequently been used as the starting point in a variety of reductions (see e.g. \cite{AAMMW11,BKS13,DLS13}).

Goldreich \cite{Gol00} used the average case hardness of random CSPs towards constructing candidate one-way functions.  Subsequently, numerous cryptographic constructions have appealed to the hardness of refutation of random CSPs \cite{ABW10,applebaum2012dichotomy,applebaum2010public}. 
Most recently, Daniely and coauthors have shown computational intractability of some long-standing problems in learning such as learning DNFs and learning halfspaces with noise, starting with the hardness of refuting random CSPs \cite{DLSS14, D16}.
We refer the reader to \cite{kothari2017sum} for a detailed survey of the implications stemming from the computational intractability of random CSPs.

An important line of work has established essentially optimal lower bounds on the size of Sum-of-Squares SDP relaxations for refuting random CSPs.
Grigoriev \cite{Gri01} and Schoenebeck \cite{Sch08} rule out efficient weak refutations for random $k$-XOR and random $k$-SAT at densities significantly smaller than $m/n < n^{\frac{k-2}{2}}$.
Specifically, Schonebeck's result implies that with high probability over $k$-XOR instances $\Phi$ with clause density $m/n < O(n^{(\frac{k-2}{2})(1-\delta)})$, the \sos hierarchy cannot even weakly refute $\Phi$ at degree $O(n^{\delta})$.
More generally, Kothari \etal \cite{kothari2017sum} showed that for a predicate $P$ which supports a $(t-1)$-wise uniform probability distribution, the corresponding random CSP requires degree $\Omega(\frac{n}{\Delta^{2/(t-2)} \log \Delta})$ sum-of-squares proofs to weakly refute.  
In other words, the sum-of-squares SDP hierarchy takes super-polynomial time for all constraint densities less than $\tilde{\Omega}(n^{(t-2)/2})$, almost matching the upper bounds in \cite{RaghavendraRS16}.

The first lower bounds for LP extended formulations were proved by Yannakakis \cite{Yannakakis91,Yannakakis88}. He showed exponential lower bounds on the size of {\it symmetric} linear programs for the Travelling Salesman Problem (TSP) and Matching problems.
Three decades later, Fiorini \etal \cite{FioriniKaibelPashkovichTheis11,FMPTW12}, showed how to avoid the symmetry assumption and show a lower bound for general assymmetric linear programming relaxations of the TSP polytope.  
This was followed by a flurry of recent work on lower bounds for extended formulations for various problems.
Rothvo{\ss} showed an exponential lower bound on the size of extended formulations for the Matching problem \cite{rothvoss2013matching}.
For CSPs, the results of \cite{chan2013approximate, KothariMR17} and \cite{DBLP:conf/stoc/LeeRS15} give lower bounds for LP and SDP extended formulations respectively.

%
%In a pioneering work, Chv\'{a}tal and Szemer\'{e}di  \cite{CS88} showed that a random $k$-SAT instance with clause density $\Delta > c$ (for $c$ constant) incurs refutations of exponential size in the refutation proof system.  Ben Sasson \etal \cite{ben2001short} strengthened this result by showing that resolution proofs refuting random $k$-SAT require width $n/\Delta^{1/(k-2)+\epsilon}$.
%
%Exponential lower bounds are also known in the Polynomial Calculus proof system \cite{alekhnovich2001lower,ben1999random}.
%

%

\section{Preliminaries}
\label{sec:prelims}
\paragraph{Notation}
We use $[n]$ to denote the set $\{1,\dots,n\}$. We will use $S^c$ to denote the complement of the set $S$.
We use $\Pr_{\D}\Brac{X}$, $\Ex[\D]{X}$ and $\Var_{\D}\Brac{X}$ to denote respectively the probability, variance and expectation of $X$ over the distribution $\D$.
For a set $\Omega$ and distribution $\D$ on $\Omega$ we write $L^2(\Omega,\D)$ to denote the space of functions $f$ on $\Omega$ with inner product given by $\inner{f,g} = \Ex[x \sim \D]{f(x)g(x)}$. When $\D$ is the uniform distribution on $\Omega$ we will denote this space simply by $L^2(\Omega)$. We will use $\pi_p$ to denote the distribution on $\{-1,1\}$ which assigns probability $p$ to $-1$ and probability $q \defeq 1 - p$ to $1$. We will use $\Ind(x = a)$ to denote the function of $x$ that is equal to one if $x = a$ and zero otherwise. For $x \in \{-1,1\}^n$ and a subset $S$ we let $\chi_S(x) = \prod_{i \in S}x_i$ be the Fourier character corresponding to $S$. We use $\ovec$ to denote the vector of all ones.

\paragraph{Relaxations of Optimization Problems}
We will be primarily interested in optimization problems over discrete domains. That is, we will focus on problems where the goal is to find the optimum value of some function on some finite set. We begin by introducing our model for maximization problems.
\begin{definition}
	
A \emph{maximization problem} \(\cP = (\cS, \cF)\)
consists of a finite set \(\cS\) of feasible solutions and a finite set \(\cF\) of nonnegative objective functions.
We define $\opt_{\cP}(f) \defeq \max_{s \in \cS}f(s)$ to be the maximum value of the objective function $f$ on $\cS$.
\end{definition}
The objective functions $f\in \cF$ should be thought of as encoding different instances of the problem. For example, for the problem of finding the maximum cut in a weighted graph, each function in $\cF$ corresponds to a different setting of the edge weights. 
\subsection{Linear Programming Formulations}
In this section we define a framework for linear programming formulations and show that a small linear programming formulation implies the existence of a small non-negative representation for a given problem.
We begin with the definition of a linear programming
formulation of a maximization problem obtained by \emph{linearizing} the objective functions and possible solutions.
\begin{definition}[LP formulation for $\cP$] 
	\label{def:lpformulation}
	Let $\cP = (\cS,
	\cF)$ be a maximization problem.
	A
	\emph{linear programming formulation of $\cP$}
	of size $R$
	consists of a linear map $A \colon \R^k \rightarrow \R^R$
	and $b \in \R^R$
	together with
	\begin{enumerate}
		\item \emph{Feasible solutions:}
		a $y^s \in \R^{k}$ with $Ay^s \leq b$
		for all $s \in \cS$, i.e., the polytope
		$\{y \in \R^{k} \mid Ay \leq b\}$
		contains the points $\{y^s \mid s \in \cS\}$,
		\item \emph{Objective functions:}
		a vector $w^f \in \R^{k}$
		satisfying
		$\inner{w^f,y^s} = f(s)$ 
		for all $f \in \cF$ and all $s \in \cS$,
		i.e., the linearizations are exact on solutions.
	\end{enumerate}
\end{definition}
An LP formulation for a problem $\cP$ naturally gives rise to the following linear programming relaxation for $\cP$:
\[
\LP(f) = \max_{y \colon Ay \leq b} \inner{w^f,y}.
\]
By the above definition, we have for each objective function $f \in \cF$
\begin{align*}
\opt_{\cP}(f) = \max_{s \in \cS}f(s)
= \max_{s \in \cS}\inner{w^f,y^s}
\leq \max_{y \colon Ay \leq b}\inner{w^f,y^s} = \LP(f).
\end{align*}
Thus, $\opt_{\cP}(f) \leq \LP(f)$, i.e. $\LP(f)$ is a relaxation of the maximization problem $\cP$.
This implies that an upper bound on the value of an LP relaxation provides an upper bound on the optimum of the original maximization problem.
More formally, we make the following definition:
\begin{definition}
	Let $c \geq 0$. An LP formulation of a maximization problem $\cP$ \emph{certifies an upper bound $c$} on $f$ if $\LP(f) \leq c$.
\end{definition}
A standard application of LP duality (see e.g. \cite{chan2013approximate}) shows that a small LP formulation for a problem $\cP$ gives rise to a small non-negative representation.
\begin{lemma}[Non-negative representation of an LP formulation]
	\label{lem:nonnegativerep}
	If a maximization problem $\cP = (\cS,\cF)$ has a linear programming formulation of size $R$ certifying upper bound $c$ on $f \in \cF$, then there exist real-valued, non-negative functions $p_i(f)$, $q_i(s) \geq 0$ for $i \in \{0\dots R\}$ such that
	\[
	c - f(s) = \sum_{i=0}^{R} p_i(f)q_i(s).
	\]
\end{lemma}
%\begin{proof}
%	An LP formulation certifying an upper bound $c$ on $f$ implies the existence of $A$,$b$ and $w^f$ such that 
%	\[
%	\max_{y \colon Ay\leq b} \inner{w^f,y} \leq c.
%	\]
%	In particular $c - \inner{w^f,y}$ is a non-negative function of $y$ on the feasible region $\{y \mid Ay \leq b\}$.
%	By linear programming duality, this implies that we can write $c - \inner{w^f,y}$ as a non-negative combination of the $R$ constraint functions $b_i - \inner{A_i,y}$, where $A_i$ is the $i$-th row of $A$. In particular, there exist non-negative numbers $p_i(f)$ such that
%	\[
%	c - \inner{w^f,y} = p_0(f) + \sum_{i=1}^R p_i(f)(b_i - \inner{A_i,y}).
%	\]
%	Again by the definition of LP formulation there exist $y^s$ such that $\inner{w^f,y^s} = f(s)$ for all $s$. Plugging this into the above equation yields
%	\[
%	c - f(s) = p_0(f) + \sum_{i=1}^R p_i(f)(b_i - \inner{A_i,y^s}).
%	\]
%	Now define $q_i(s) \defeq b_i - \inner{A_i,y^s}$ for $i \in \{1\dots R\}$ and $q_0(s) \equiv 1$. Then the above equation simplifies to
%	\[
%	c - f(s) = \sum_{i=0}^{R} p_i(f)q_i(s).
%	\]
%\end{proof}

\subsection{Random Instances of Constraint Satisfaction Problems}
In this section we formally introduce the class of constraint satisfaction problems and our model for random instances.
Specifically, we will consider CSPs defined by a single boolean predicate $P$ applied to $k$-tuples of literals (i.e. possibly negated variables) from a set of $n$ variables.
\begin{definition}
	\label{def:CSP}
	A predicate $P\from \{-1,1\}^k \to \{0,1\}$ defines a constraint satisfaction problem $\CSP(P)$. An instance $I$ of $\CSP(P)$ is given by a set of $n$ variables $x_1,\dots,x_n$ and a set of $m$ \emph{constraints} $(S,b)$, where $S = (S_1,\dots,S_k)$ is a \emph{scope} of $k$ distinct variable indices and $b\in\{-1,1\}^k$ is a string of $k$ negations. The objective is: given an instance $I$, find an assignment of values from $\{-1,1\}^n$ to the variables $x_1,\dots,x_n$ in order to maximize
	\[
	f_I(x) \defeq \sum_{(S,b)} P(b_{1}x_{S_1},\dots,b_{k}x_{S_k}).
	\]
	We will use $\opt_P(I)$ to denote this maximum.
\end{definition}
Note that $\CSP(P)$ is a maximization problem where the set of possible solutions is $\cS = \{-1,1\}^n$ and the set of objective functions $\cF$ is the set of all $f_I$ for each instance $I$ of $\CSP(P)$. We will also use $\Enk$ to denote the set of possible scopes i.e. the set of $k$-tuples of distinct indices in $[n]^k$. For a subset $U \subseteq \Enk$ we will use $I_U$ to denote the restriction of the instance $I$ to the set of scopes and variables containted in $U$.

We say that a predicate $P$ supports a $(t-1)$-wise uniform distribution on satisfying assignments if there is a probability distribution $\eta_P$ with the following properties:
\begin{itemize}
	\item $\eta_P$ is supported on $P^{-1}(1)$.
	\item For every set $x_S$ of $(t-1)$ input variables, the distribution of $x_S$ is uniform.
\end{itemize}
Note, for example, that the uniform distribution on satisfying assignments to $\maxkxor$ is $(k-1)$-wise uniform.

One natural distribution on random instances of CSPs is given by choosing $\Delta n$ constraints $(S,b)$ independently and uniformly at random, where $\Delta \geq 0$ is called the \emph{constraint density}. For technical reasons we will use a slightly different distribution in our analysis. However, our results for this modified distribution can be translated into similar results for the original distribution with density $\Delta$.
\begin{definition}
	\label{def:randomCSP}
	Given a parameter $p \geq 0$ a random instance $I$ of $\CSP(P)$ on $n$ variables is sampled as follows:
	\begin{itemize}
		\item Include each constraint scope $S$ independently with probability $p$.
		\item For each $S$ chosen, sample a uniformly random string $b_S \in \{-1,1\}^k$.
		\item Let $I$ consist of all the constraints $(S,b_S)$ chosen by this process.
	\end{itemize}
	We will use $\D(p)$ to denote this distribution.
\end{definition}
In short, each constraint scope is included independently with probability $p$ and each scope is assigned one uniform random string of negations.
By analogy with the situation where the number of constraints is fixed to $\Delta n$, we will set $\Delta \defeq \frac{p \abs{\Enk}}{n}$ so that the expected number of constraints in an instance sampled from $\D(p)$ is $p \abs{\Enk} = \Delta n$.
Standard Chernoff bounds imply that the number of constraints chosen is close to $\Delta n$ with high probability.
\begin{lemma}
	\label{lem:constraintconcentration}
	Let $m$ be the number of constraints in an instance sampled from $\D(p)$.
	\[
	\Pr_{\D(p)}\Brac{\abs{m - \Delta n} > \eps \Delta n} \leq 2\exp{\left(-\frac{\eps^2\Delta n}{3}\right)}.
	\]
\end{lemma}
%\begin{proof}
%	Every constraint scope is included independently with probability $p$ so $m$ is the sum of $\abs{\Enk}$ independent Bernoulli random variables. Thus
%	\begin{align*}
%	\Pr_{\D(p)}\Brac{ \abs{m - \Delta n} > \eps \Delta n} &\leq 2\exp{(-\frac{\eps^2\Delta n}{3} )}.
%	\end{align*}
%\end{proof}

\subsection{Strong Refutation}
For a random CSP there are constants $\Delta_0 < \Delta_1$, such that a random instance with constraint density $\Delta < \Delta_0$ is satisfiable with high probability, and for $\Delta > \Delta_1$ is unsatisfiable with high probability. Further, for some CSPs including $k$-SAT, it is conjectured that these two numbers are equal, and there is a sharp satifiability threshold $\Delta_c = \Delta_0 = \Delta_1$.
For random instances $I$ with density $\Delta > \Delta_1$ a natural computational problem is to efficiently find a certificate that $I$ is unsatisfiable. This is known as \emph{refuting} (or sometimes \emph{weakly refuting}) the instance $I$. Our results apply to the computationally harder problem of certifying an upper bound on the maximum number of constraints that can simultaneously be satisfied in the instance $I$.

\begin{definition}
	Let $A$ be an algorithm that takes as input an instance $I$ of $\CSP(P)$ and outputs a number. We say that $A$ is a \emph{strong refutation} algorithm for $\CSP(P)$ if $A$ always outputs a valid upper bound on $\opt_P(I)$. If additionally $A$ outputs $1-\delta$ with probability $s$ over the random choice of $I$ sampled from a distribution $\mu$, then we say that $A$ is a $\delta$-strong refutation algorithm for $\mu$ with success probability $s$.
\end{definition}
By \prettyref{def:lpformulation}, any LP formulation for $\CSP(P)$ immediately gives rise to a strong refutation algorithm. For our purposes, let $\Delta_*$ be the density for $\CSP(P)$ such that a random instance sampled from $\D(p)$ with $p > (1 + \eps)\frac{n\Delta_*}{\Abs{\Enk}}$ is not even $1-\delta$ satisfiable (for some constant $\delta$) with high probability. One can take $\Delta_*$ to be the first moment upper bound on the satisfiability threshold. It is natural to ask if it is possible to construct an LP formulation that $\delta$-strongly refutes such CSPs with any reasonable probability. Our main result shows that such an LP formulation must have sub-exponential size, where the exponent varies depending on the parameter $p$.

\subsection{Blockwise-Dense Distributions}
A key element of our proof lies in being able to approximately decompose any distribution on instances into a convex combination of ``simpler'' distributions. The class of simple distributions we will consider are based on the following definitions. These definitions appeared first (in a slightly different form) in \cite{GoosLMWZ16} in the context of communication complexity, and then were later used to prove extended formulation lower bounds in \cite{KothariMR17}.
\begin{definition}
	Let $\D$ be a distribution on instances $I$. We say that $\D$ is \emph{blockwise-dense} with parameter $\delta$ relative to $\D(p)$ if 
	for every subset $V \subseteq \Enk$ and every instance $I'_V$ we have 
	\[
	\Pr_{I \sim \D}\Brac{I_V = I'_V} \leq \Pr_{I \sim \D(p)}\Brac{I_V = I'_V}^{1 - \delta}
	\]
\end{definition}
That is, a blockwise-dense distribution $\D$ does not assign a much higher probability to any instance than $\D(p)$ does. The exponent of $1-\delta$ should be thought of as some constant with a small value of $\delta >0$. Unfortunately, we will not be able to decompose any distribution into a convex combination of blockwise-dense distributions.
In particular, a distribution which fixes a set of $d$ scopes and their corresponding negations cannot be decomposed into blockwise-dense distributions if $d$ is large enough. This example motivates the following definition.
\begin{definition}
	Let $\D$ be a distribution on instances $I$. We say that $\D$ is $d$\emph{-conjunctive blockwise-dense} ($d$-CBD) with parameter $\delta$ relative to $\D(p)$ if
	\begin{enumerate}
		\item There exists $U \subseteq \Enk$ with $\abs{U} \leq d$ and a restricted instance $I^*_U$ such that $\Pr_{I \sim \D}\Brac{I_U = I^*_U} = 1$.
		\item For every subset $V \subseteq \Enk \setminus U$ and every instance $I'_V$ we have 
		\[
		\Pr_{I \sim \D}\Brac{I_V = I'_V} \leq \Pr_{I \sim \D(p)}\Brac{I_V = I'_V}^{1 - \delta}
		\]
	\end{enumerate}
	We call the coordinates in the subset $U$ the \emph{fixed} coordinates of $\D$.
\end{definition}
The basic intuition is that $d$-CBD distributions fix some small set of $d$ coordinates and then sample all other coordinates with probability not too much higher than the probability assigned to them by $\D(p)$. In \prettyref{sec:cbddecomposition} we will show that any distribution can be decomposed into a convex combination of conjunctive blockwise-dense distributions plus an error set which has small measure under $\D(p)$.

%\subsection{Sherali-Adams Pseudo-densities}
%The Sherali-Adams linear programming hierarchy is a sequence of LP relaxations of increasing size. Intuitively, the $d$-th linear program in this sequence constructs locally defined probability distributions on $d$-tuples of variables which locally satisfy the constraints.
%This collection of local distributions gives rise to an object called a pseudo-density, because it ``acts like'' a true probability density on small collections of variables. We now give the formal definition for this object.
%\begin{definition}
%	A degree-$d$ Sherali-Adams pseudo-density $H:\{-1,1\}^n \to \R^+$ for an instance $I = (y,b)$ is a non-negative function satisfying:
%	\begin{itemize}
%		\item For every non-negative $d$-Junta $J(x)$ (i.e. function depending on only $d$ input variables), $\Ex[x]{H(x)J(x)}\geq 0$.
%		\item $\Ex[x]{H(x)}$ = 1.
%		\item For every constraint $C$, $\Ex[x]{H(x)\Ind(C(x) = 1)} = 1$.
%	\end{itemize}
%\end{definition} 

\section{Proof Overview}
\label{sec:proofoverview}
In this section we introduce our approach to proving LP lower bounds using pseudo-calibration. The main idea is that we want to use the small non-negative representation given by \prettyref{lem:nonnegativerep} to derive a contradiction. In particular, we want to find a function $H(x,I)$ that witnesses a violation of the equality $c - f_I(x) = \sum_i p_i(I)q_i(x)$. That is, we want both
\begin{equation}
\label{eqn:proofapproachlhs}
\E[H(x,I)(c - f_I(x))] < 0
\end{equation}
and
\begin{equation}
\label{eqn:proofapproachrhs}
\E\left[H(x,I) \sum_{i=0}^{R} p_i(I)q_i(x)\right] \geq 0
\end{equation}
where the expectation is taken over random CSPs conditioned on being in the set of instances $I$ on which the LP succeeds.

One natural way to try to construct such an $H$ is to first choose a \emph{planted} distribution $\D_*$ on pairs $(x,I)$ of assignments and instances, so that $x$ always completely satisfies the instance $I$. Next, let $\mu_*(x,I)$ be the density of $\D_*$ relative to the distribution where $I \sim \D(p)$ and $x \sim \{-1,1\}^n$ independently. Since $x$ always satisfies $I$ we have
\[
\Ex[\substack{I \sim \D(p)\\ x \sim \{-1,1\}^n}]{\mu_*(x,I)(c - f_I(x))} = \Ex[(x,I) \sim \D_*]{c - f_I(x)}< 0
\]
for any $c$ which is sufficiently less than the expected number of constraints in $I$. Furthermore, since each of the $q_i(x)$ and $p_i(I)$ are non-negative we have by linearity of expectation
\[
\Ex[\substack{I \sim \D(p)\\ x \sim \{-1,1\}^n}]{\mu_*(x,I)\sum_{i=0}^R q_i(x)p_i(I)} = \sum_{i=0}^{R}\Ex[(x,I) \sim \D_*]{q_i(x)p_i(I)} \geq 0.
\]

Of course this cannot be a correct proof of a lower bound, because we have not yet used that the LP formulation has size $R$ at all. In particular, the distribution $\D_*$ has very small probability mass under the background measure $\D(p)$, and so $\mu_*$ is supported only a small fraction of instances.
Furthermore, an LP--or any algorithm--for refuting a random CSP is not supposed to work on planted instances. So the fact that $\mu_*$ is identically zero whenever the above non-negative representation exists is unsurprising.

The solution to this problem is a technique called pseudo-calibration, which allows us to construct a function $H(x,I)$ with similar properties to $\mu_*(x,I)$, but which is much more spread out over the space of possible instances. Briefly, $H(x,I)$ is constructed by simply dropping the high-degree terms from the Fourier expansion of $\mu_*$ over an appropriately chosen Fourier basis. Our goal is then to show that the two inequalities \pref{eqn:proofapproachlhs} and \pref{eqn:proofapproachrhs} are satisfied. 

The inequality in \pref{eqn:proofapproachlhs} is easier to satisfy since the objective function $f_I(x)$ is low-degree, and so when averaging with a low-degree trunctation of $\mu_*(x,I)$ nothing changes. The inequality in \pref{eqn:proofapproachrhs} is more difficult, and requires the bulk of the work in the paper. The issue is that the $p_i$ and $q_i$ may be arbitrary non-negative functions. 
The only fact we have to work with is that there are at most $R$ of them. 

The main idea is to look at each $p_i(I)$ individually, and to think of each one (after re-normalization) as a probability distribution $\D_i$ over instances $I$. We then decompose each $\D_i$ into a collection of simpler (conjunctive blockwise dense) distributions $\D_{i,j}$ plus a small error term. 
Since there are at most $R$ distributions $\D_i$, if we make the error term small enough (say having total probability mass much less than $\frac{1}{R^2}$) then the total contribution of the error will be negligible.

Finally, using Fourier analytic properties of $H(x,I)$, we show that $\Ex[I \sim \D_{i,j}]{H(x,I)}$ is non-negative with high-probability over $x$ for all of the simpler distributions $\D_{i,j}$. Thus, since each $q_i(x)$ is non-negative, we are able to conclude that \pref{eqn:proofapproachrhs} holds, up to the error introduced by our decomposition into simpler distributions. This yields the desired contradiction.
We now give a more detailed overview of the techniques involved in each piece of the above high-level description.

\paragraph{Pseudo-calibration and the Fourier expansion of $\mu_*$}
As described above we define the function $H(x,I)$ by dropping the high degree terms in the Fourier expansion of $\mu_*(x,I)$. We then wish to show that, for any conjunctive blockwise dense distribution $\D$, the function $G(x) = \Ex[I \sim \D]{H(x,I)}$ is non-negative with high probability over $x$. We will achieve this by showing that the Fourier coefficients of $G$ decay very rapidly as the degree increases. We then appeal to a concentration inequality based on the hypercontractivity of bounded degree polynomials over $\{-1,1\}^n$.

Since the CBD distribution $\D$ fixes some subset $U$ of the scopes of $I$, the key to bounding the Fourier coefficients of $G$ is to bound the Fourier coefficients of $H$ after fixing $U$. If $H$ were an actual probability density, this could be achieved by bounding the coefficients of the density conditioned on $U$. However, $H$ is merely the low-degree terms in the Fourier expansion of the true probability density $\mu_*$, so this cannot be done directly.

Instead, let $\mu_*\rvert_{U}$ and $H\rvert_U$ respectively denote the density of $\mu_*$ conditioned on $U$ and the function $H$ with the scopes in $U$ fixed. Then we have four functions to consider $\mu_*$, $\mu_*\rvert_U$, $H$ and $H\rvert_U$. Since $H$ is a low-degree truncation of $\mu_*$ it is easy to relate these two functions. Similarly, we can relate $\mu_*$ to $\mu_*\rvert_U$ and $H$ to $H\rvert_U$. 
This allows us to write down an explicit formula for the difference $\mu*\rvert_U - H\rvert_U$ in terms of $\mu_*$. Thus to compute bounds on the Fourier coefficients of $H\rvert_U$ we need only to compute bounds on the Fourier coefficients of $\mu_*\rvert_U$.

Up to this point, our argument is completely generic in that it is not specific to CSPs, and could hypothetically be carried out for other combinatorial optimization problems. In fact, the only property we have used so far is that the uniform distribution on instances and solutions is a product measure, so that there is some reasonable Fourier basis in which to express $\mu_*$. To actually compute the requisite bounds on the Fourier coefficients, we of course must use some facts particular to the planted distribution $\D_*$ for CSPs. In particular, we compute explicit bounds on the Fourier coefficients of $\mu_*$ and $\mu_*\rvert_U$ for any CSP supporting a $(t-1)$-wise uniform distribution on satisfying assignments.

It is worthwhile to note that this approach does not involve constructing a valid Sherali-Adams integrality gap for any one fixed instance (as is done in prior work). Rather, the proof is purely Fourier-analytic, and relies only on showing that certain functions are bounded \emph{on average} over instances.

%Instead we are able to show that, after fixing $U$, the low-degree part of $H$ corresponds to that of $\mu_*$ conditioned on $U$. However, the high-degree Fourier coefficients of $H$ will in general be much larger than those of $\mu_*$. We resolve this by showing that, if the Fourier coefficients of the density of $\mu_*$ conditioned on $U$ decay rapidly enough, this decay can offset the loss in the high-degree part of $H$. The final ingredient is then to actually bound the Fourier coefficients of $\mu_*$ conditioned on $U$ and show that they achieve the level of decay required.

\paragraph{Conjunctive blockwise dense decompositions}
In order to decompose a distribution $\D$ into a convex combination of conjunctive blockwise dense distributions, we modify the approach of \cite{GoosLMWZ16, KothariMR17}. 
These prior works start with some distribution which does not put too much probability mass on any one instance, and then uses a greedy strategy to decompose any such distribution.
We cannot use these results directly because they give a dependence between the maximum size $B$ of a fixed block $U$ in the resulting CBD decompositions, and the $\epsilon$ probability mass of $\D$ left over as error. This dependence means that we cannot make $\epsilon$ small enough for our applications without having $B$ grow to be larger than the total number of coordinates available to put in the fixed block $U$.

Therefore we modify their approach by using essentially the same greedy decomposition algorithm, but also recursively truncating the intermediate distributions to remove the parts which put too much mass on any one instance relative to the background measure.
We ensure that the set of instances removed in this way has small mass under the background measure, and this allows us to achieve the required dependence between $B$ and $\epsilon$.
%The decomposition works by first truncating $\D$ by removing all instances $I$ which occur with very high probability (depending on $R$) relative to their probability in $\D(p)$. Next we find a maximal set of scopes $V$ such that the restricted instance $I_V$ takes on some value $I^*_V$ with a large enough probability to witness the fact that $\D$ is not conjunctive blockwise dense. We then split the distribution $\D$ by conditioning on whether or not $I_V$ takes the value $I^*_V$.

%Using the maximality of $V$ we are able to conclude that the  distribution $\D$ conditioned on $I_V = I^*_V$ is conjunctive blockwise dense. For the piece where $\D$ is conditioned on $I_V \neq I^*_V$ we recursively apply the entire process described above. We set the parameters for the truncation and conditioning steps so that the total probability mass removed (with respect to the background measure $\D(p)$) is small compared to $R$. Our approach deviates from prior decomposition strategies in that we truncate the distribution at each recursive step, since the conditioning may increase the probability of some instances. 

\paragraph{Outline of the rest of the paper}
In \pref{sec:pseudocalibration} we formally define the pseudo-calibration of a planted density $\mu_*(x,I)$. We then identify Fourier analytic properties of $\mu_*$ which guarantee that, after averaging $\mu_*$ over a conjunctive blockwise dense distribution on instances, the resulting function of $x$ is non-negative with high probability.
In \pref{sec:csppseudocalibration} we compute bounds on the Fourier coefficients of $\mu_*$ for the case of constraint satisfaction problems. In particular, we show that the Fourier analytic properties identified in \pref{sec:pseudocalibration} are satisfied for CSPs which support a $(t-1)$-wise independent distribution on satisfying assignments to their defining predicate.
In \pref{sec:cbddecomposition} we prove that any distribution on instances can be decomposed into a convex combination of conjunctive blockwise dense distributions plus a small error term.
Finally in \pref{sec:lbproof} we put all the pieces together to complete the proof as described above.

\section{Pseudo-Calibration}
\label{sec:pseudocalibration}
In this section we formally define the pseudo-calibration of the planted density for an optimization problem. We then introduce generic Fourier analytic properties of the planted density that are sufficient for proving LP lower bounds. In particular, we will show in this section that any planted density satisfying these properties is non-negative with high probability after averaging over a conjunctive blockwise dense distribution on instances.

\subsection{The Planted Density}
The first step to formally define pseudo-calibration is to define what is meant by a planted distribution $\D_*$. Let $\Omega$ be a finite set and $M$ be an integer depending on $n$. We will assume that instances $I$ of the optimization problem are elements of $\Omega^M$ and there is some null distribution on instances denoted $\D_\emptyset$.
\begin{definition}
	\label{def:planted}
	A planted distribution $\D_*$ is a distribution on solution and instance pairs $(x,I)$ which can be sampled as follows:
	\begin{itemize}
		\item Sample $x$ uniformly at random from $\{-1,1\}^n$.
		\item Sample $I$ from a distribution $\D_x$ such that all instances $I \in \supp\{D_x\}$ satisfy
		\[x = \argmax_{x' \in \{-1,1\}^n} f_I(x').\]
	\end{itemize}
\end{definition}
In words, the planted distribution $\D_*$ first samples a uniformly random solution $x$, and then samples an instance $I$ so that $x$ is an optimal solution for $I$.
Let $\mu_*(x,I)$ be the density of $\D_*$ relative to the distribution $\{-1,1\}^n \times \D_{\emptyset}$ where solutions and instances are sampled independently.

Next we describe the appropriate Fourier basis for expanding $\mu_*$. Let $\chi_\alpha$ be the Fourier basis for $L^2(\{-1,1\}^n)$. That is, for $\alpha \subseteq [n]$
\[
\chi_\alpha(x) = \prod_{i \in \alpha} x_i.
\]
Let $\zeta_\sigma$ be a Fourier basis for $L^2(\Omega^M,\D_\emptyset)$. That is, for $\sigma \in \{0,\dots,\abs{\Omega}-1\}^M$
\[
	\zeta_\sigma(I) = \prod_j \zeta_{\sigma_j}(I_j)
\]
where we require that
\begin{itemize}
	\item $\zeta_0 \equiv 1$.
	\item $\Ex[\D_{\emptyset}]{\zeta_i\zeta_j} = 0$ for $i \neq j$.
	\item $\Ex[\D_{\emptyset}]{\zeta_j^2} = 1$ for all $j$.
\end{itemize}
We also define the support of $\sigma$ to be $\supp(\sigma) \defeq \{j \in M \mid \sigma_j \neq 0\}$ and set $\abs{\sigma} = \Abs{\supp(\sigma)}$.
For a pair $\sigma,\sigma'$ with disjoint support we will write $\sigma + \sigma'$ to denote the coordinate-wise sum of $\sigma$ and $\sigma'$. In particular, $\sigma + \sigma'$ is equal to $\sigma$ on $\supp(\sigma)$, equal to $\sigma'$ on $\supp(\sigma')$, and equal to zero everywhere else.

We can then write $\mu_*(x,I)$ as a function in the tensor product of the spaces spanned by the $\chi_\alpha$ and $\zeta_\sigma$.
That is
\[
	\mu_*(x,I) = \sum_{\alpha,\sigma} \widehat{\mu_*}(\alpha,\sigma) \chi_\alpha(x)\zeta_\sigma(I)
\]
where the Fourier coefficients above are given by the inversion formula
\[
	\widehat{\mu_*}(\alpha,\sigma) = \Ex[(x,I) \sim \{-1,1\}^n \times \D_\emptyset]{\mu_*(x,I)\chi_\alpha(x)\zeta_\sigma(I)} = \Ex[(x,I) \sim \D_*]{\chi_\alpha(x)\zeta_\sigma(I)}.
\]

Finally, we are ready to define the pseudo-calibrated density. For a pair $d = (d_x,d_I) \in \N^2$ let $A(d) \defeq \{(\alpha,\sigma) \mid \abs{\alpha} \leq d_x, \abs{\sigma} \leq d_I\}$. That is, $A(d)$ corresponds to the set of Fourier coefficients where $x$ has degree at most $d_x$ and $I$ has degree at most $d_I$.
We will use $L_d$ to denote the linear projection operator onto the span of $\{\chi_\alpha \zeta_\sigma\}_{(\alpha,\sigma) \in I(d)}$.
\begin{definition}
\label{def:pseudocalibration}
	For $d = (d_x,d_I)\in \N^2$ we define the $d$-pseudo-calibration $\bar{\mu}_* \defeq L_d \mu_*$ . Equivalently
	\[
		\bar{\mu}_*(x,I) = \sum_{(\alpha,\sigma)\in A(d)} \widehat{\mu_*}(\alpha,\sigma) \chi_\alpha(x)\zeta_\sigma(I)
	\]
\end{definition}

\subsection{Conditioning the Pseudo-calibrated Density}
In this section we will identify the properties of a pseudo-calibrated density that we will use for proving LP lower bounds. In particular, we will show that these properties imply that, after averaging $\bar{\mu}_*$ over a CBD distribution on instances, the resulting function is non-negative with high probability.
Since CBD distributions fix a block $U$ of coordinates of $I$ to a given value, the analysis will be based on computing bounds on the Fourier coefficients of $\bar{\mu}_*$ after fixing such a block $U$. In particular, we will relate this to the Fourier coefficients of $\mu_*\rvert_U$, the density of $\D_*$ conditioned on the fixed block $U$.

Recall that for a fixed block $U\subseteq [M]$ we defined the restricted instance $I_U$ to be the subset of coordinates of $I$ corresponding to $U$.
We will also define $V(U)$ to be the set of coordinates for variables contained in the constraint scopes of $I_U$. When $U$ is clear from context we will write $x_V\defeq x_{V(U)}$ for the variables on this set of coordinates.
We now introduce notation for the restriction of a function on instances.

\begin{definition}
\label{def:restrictplanted}
Let $f:\{-1,1\}^n\times \Omega^M \to \R$ be a function on assignments $x$ and instances $I$. Let $U \subseteq [M]$ and fix a restricted instance $I_U$. Then $R_U f$ is the function defined on the coordinates outside of $U$ given by
\[
	R_U f(x,I_{\setcomp{U}}) \defeq f\left(x,(I_{\setcomp{U}},I_U)\right).
\]
\end{definition}
A key fact about the restriction operator $R_U$ is that it transforms the Fourier coefficients of a function in a simple way.
\begin{observation}
	\label{obs:fourierrestrict}
	Let $f(x,I)$ be a function on instances and assignments. Then for all $\sigma$ with $\supp(\sigma) \cap U = \emptyset$ 
	\[
		\widehat{R_U f}(\alpha,\sigma) = \sum_{ \sigma' : \supp(\sigma') \subseteq U} \widehat{f}(\alpha, \sigma + \sigma')\zeta_{\sigma'}(I_U).
	\]
\end{observation}
\begin{proof}
	The proof is immediate from the fact that precisely those coordinates in $U$ are set to the fixed values $I_U$, and all other coordinates remain as inputs to $R_U f$.
\end{proof}

Next we formally define the conditional planted density. Recall that $V = V(U)$ is the set of coordinates for variables in $x$ corresponding to the restricted instance $I_U$.
\begin{definition}
\label{def:condtionplanted}
Let $U \subseteq [M]$ and fix a restricted instance $I_U$. Then $\mu_*\rvert_U$ is the planted density conditioned on $U$ defined by
\[
	\mu_*\rvert_U\left(x,I_{\setcomp{U}}\right) \defeq \frac{\Pr_{\D_*}[x_{\setcomp{V}},I_{\setcomp{U}} \mid x_V, I_U]}{\Pr_{\D_\emptyset}[x_{\setcomp{V}},I_{\setcomp{U}} \mid x_V, I_U]}
	= \frac{\Pr_{\D_*}[x_{\setcomp{V}},I_{\setcomp{U}} \mid x_V, I_U]}{\Pr_{\D_\emptyset}[x_{\setcomp{V}},I_{\setcomp{U}}]}
\]
\end{definition}
An important but subtle point in the above definition is that we think of $\mu_*\rvert_U$ as a function of all the variables in $x$, including those variables $x_V$ on which we condition. 
We are now ready to state the properties of the planted density that allow us to prove LP lower bounds.

\begin{definition}
\label{def:plantedproperties}
	Fix $d_x, d_I$ and $\delta > 0$.
	The planted density $\mu_*$ exhibits $\epsilon(s_x,s_I)$-\emph{conditional Fourier decay} if there is a bound $\Abs{\widehat{\mu_*\rvert_U}(\alpha,\sigma)} \leq B_U(\alpha,\sigma)$, such that for all $b$-CBD distributions $\D$ with parameter $\delta$ and fixed block $U$:
	\begin{enumerate}
		\item For each $\alpha \subseteq [n]$ with $\abs{\alpha} = s_x \leq d_x$ and each $s_I\leq d_I$
		\[
			\sum_{s_I \leq \abs{\sigma} \leq d_I} B_U(\alpha,\sigma)\Ex[\D]{\zeta_\sigma} \leq \binom{n}{s_x}^{-\frac{1}{2}}\epsilon(s_x, s_I).
		\]
		\item For each $\alpha$ and each pair $\sigma$, $\sigma'$ with disjoint support
		\[
			B_{\emptyset}(\alpha,\sigma + \sigma') \norm{\zeta_{\sigma'}}_\infty \leq  B_{\emptyset}(\alpha,\sigma)
		\]
	\end{enumerate}
\end{definition}
Note that $\epsilon(s_x,s_I)$ may depend on $\delta$ and $b$, but not on any other properties of the distribution $\D$.
The first property should be thought of as giving an upper bound on the Fourier coefficients of the conditional density after averaging over a blockwise-dense distribution on instances.
The factor of $\binom{n}{s_x}^{-\frac{1}{2}}$ is included so that the sum of the squares of these averaged Fourier coefficients is bounded by $\epsilon(s_x,s_I)^2$.
The second property allows us to obtain a bound on the Fourier coefficients of the restricted density in terms of the original density by appealing to \pref{obs:fourierrestrict}.

In the remainder of this section we will prove that the properties from \pref{def:plantedproperties} imply that, after averaging $\bar{\mu}_*$ over a CBD distribution, the resulting function is non-negative with high probability.
It is straightforward to relate the restriction of the planted density $R_U\mu_*$ to the conditional planted density $\mu_*\rvert_U$ using Bayes' rule.
However, we are actually interested in the restriction of the pseudo-calibrated density $R_U\bar{\mu}_*$.
This adds some subtlety to the argument which in turn necessitates the second condition in \pref{def:plantedproperties}.
In particular, we still use the relationship obtained via Bayes' rule, but must deal with an additional error term introduced by pseudo-calibration.
This additional error term is the one which requires the second condition mentioned above.

First, we introduce notation for the density on instances within the set of scopes $U$.

\begin{definition}
	\label{def:plantedwithinU}
	Let $U\subseteq [M]$. Then $\pi_U$ is the density on instances within $U$ given by
	\[
		\pi_U(x_V, I_U) \defeq \frac{\Pr_{D_*}[x_{V}, I_U]}{\Pr_{D_\emptyset}[x_{V}, I_U]}
	\]
	where the notation $x_V$ and $I_U$ emphasizes the fact that $\pi$ is defined only over coordinates in $U$ and $V = V(U)$.
\end{definition}

With these definitions in hand we can relate the restriction of the planted density to the conditional planted density.

\begin{lemma}
	\label{lem:restrictcondition}
	Let $U$ and $I_U$ be as in \pref{def:restrictplanted}. Then
	\[
		R_U\mu_*\left(x,I_{\setcomp{U}}\right) = \mu_*\rvert_U\left(x,I_{\setcomp{U}}\right)\pi_U(x_V, I_U).
	\]
\end{lemma}
\begin{proof}
	By \pref{def:restrictplanted} and Bayes' rule
	\[
		R_U\mu_*\left(x,I_{\setcomp{U}}\right) = \frac{\Pr_{D_*}[x,(I_{\setcomp{U}},I_U)]}{\Pr_{D_\emptyset}[x,(I_{\setcomp{U}},I_U)]}
		= \frac{\Pr_{\D_*}[x_{\setcomp{V}},I_{\setcomp{U}} \mid x_V, I_U] \Pr_{\D_*}[x_V,I_U]}{\Pr_{\D_\emptyset}[x_{\setcomp{V}},I_{\setcomp{U}} \mid x_V, I_U]\Pr_{\D_\emptyset}[x_V,I_U]}
	\]
	the proof then follows by the definitions of $\mu_*\rvert_U$ and $\pi_U$.
\end{proof}
Our goal is to use the relationship from the above lemma to obtain bounds on the Fourier coefficients of $R_U \bar{\mu}_*$. We begin by deriving a formula for $R_U\bar{\mu}_*$ in terms of $\pi_U$ and $\mu_*\rvert_U$ plus an error term.

\begin{lemma}
	\label{lem:highlowdecomp}
	Let $U$ and $I_U$ be as in \pref{def:restrictplanted}. Let $d = (d_x,d_I)$ be the degree bound so that $\bar{\mu}_* = L_d\mu_*$. Further, let $d' = (d_x,d_I - \abs{U})$ and $d'' = (d_x, d_I - 2\abs{U})$. Then
	\[
		R_U\bar{\mu}_* = \pi_U \cdot L_{d''}\mu_*\rvert_U + h
	\]
	where the error term $h$ is given by
	\[
		h = (L_{d'}R_U\mu_* - L_{d''}R_U\mu_*) + (R_U\bar{\mu}_* - L_{d'}R_U\bar{\mu}_*).
	\]
\end{lemma}
\begin{proof}
	Observe that by \pref{obs:fourierrestrict} and the definition of $d$ and $d'$
	\[
		L_{d'}R_U \bar{\mu}_* = L_{d'}R_U L_d \mu_* = L_{d'} R_U \mu_*.
	\]
	Further, by \pref{lem:restrictcondition} we have
	\[
		L_{d'}R_U \mu_* = L_{d'}(\pi_U \cdot \mu_*\rvert_U).
	\]
	Using the fact that $\pi_U$ has degree at most $\abs{U}$ in $I$, and the definition of $d''$ we obtain
	\[
		L_{d'}(\pi_U \cdot \mu_*\rvert_U) = \pi_U \cdot L_{d''}\mu_*\rvert_U + (L_{d'}R_U\mu_* - L_{d''}R_U\mu_*).
	\]
	Putting it all together we get
	\[
		L_{d'}R_U \bar{\mu}_* = \pi_U \cdot L_{d''}\mu_*\rvert_U + (L_{d'}R_U\mu_* - L_{d''}R_U\mu_*).
	\]
	Adding $R_U \bar{\mu}_*$ to both sides and rearranging completes the proof.
\end{proof}

We next use the properties from \pref{def:plantedproperties} to prove bounds on the Fourier coefficients of the error term $h$ from \pref{lem:highlowdecomp} above. The main point here is that, although $h$ may lose some factor compared to the level of decay of $\mu_*$, this can be compensated for by the fact that $h$ consists only of high degree terms in $I$.
\begin{lemma}
	\label{lem:hfourierbounds}
	Let $U$, $I_U$, $\D$ and $d$ be as in \pref{def:plantedproperties} and let $\mu_*$ exhibit $\epsilon(s_x,s_I)$-conditional Fourier decay. Let $u = \max(s_I,d_I - 2\Abs{U})$.
	For $\abs{\alpha} = s_x \leq d_x$ and $s_I \leq d_I - \Abs{U}$
		\[
			\sum_{s_I \leq \abs{\sigma} \leq d_I} \Abs{\widehat{h}(\alpha,\sigma)\Ex[\D]{\zeta_\sigma}} \leq \binom{n}{s_x}^{-\frac{1}{2}}2^{\Abs{U}}\epsilon(s_x,u).
		\]
\end{lemma}
\begin{proof}
	First observe that since $h = (L_{d'}R_U\mu_* - L_{d''}R_U\mu_*) + (R_U\bar{\mu}_* - L_{d'}R_U\bar{\mu}_*)$, every Fourier coefficient of $h$ is either a Fourier coefficient of $R_U\mu_*$ or $R_U\bar{\mu}_*$.
	Let $A = \{\sigma \mid \supp(\sigma) \subseteq U\}$.
	Then, whenever $d - \abs{U} \geq \abs{\sigma} \geq d_I - 2\abs{U}$, \pref{obs:fourierrestrict} implies
	\begin{align*}
		\Abs{\widehat{h}(\alpha,\sigma)} 
			&\leq \sum_{\sigma' \in A} \Abs{\widehat{\mu_*}(\alpha,\sigma + \sigma')\zeta_{\sigma'}(I_U)}\\
			&\leq \sum_{\sigma' \in A} B_\emptyset(\alpha,\sigma + \sigma')\norm{\zeta_{\sigma'}}_\infty\\
			&\leq 2^{\Abs{U}}B_\emptyset(\alpha,\sigma)
	\end{align*}		
	where the bound $B$ in the final inequality comes from the second property in \pref{def:plantedproperties}.
	Further, if $\abs{\sigma} < d_I - 2\abs{U}$ then $\widehat{h}(\alpha,\sigma) = 0$. Letting $J$ denote the interval $[u,d_I - \abs{U}]$ we then have
	\begin{align*}
		\sum_{\abs{\sigma} \leq d_I}\Abs{\widehat{h}(\alpha,\sigma)\Ex[D]{\zeta_\sigma}}
			\leq \sum_{\abs{\sigma} \in J}2^{\Abs{U}}B_\emptyset(\alpha,\sigma)\Ex[D]{\zeta_\sigma}
			\leq 2^{\Abs{U}}\binom{n}{s_x}^{-\frac{1}{2}}\epsilon(s_x,u)
	\end{align*}
	where the final inequality follows from the first property of \pref{def:plantedproperties} applied in the $U = \emptyset$ case.

\end{proof}
To get a better understanding of the lemma, suppose that $\epsilon(s_x,s_I)$ decays very quickly with $s_I$, say $\epsilon(s_x,s_I) = O\left(n^{-s_I}\right)$. 
Then whenever $d_I - 2\Abs{U} \gg \Abs{U}$ we will have $2^{\Abs{U}}\epsilon(s_x,u) = O\left(n\right)^{-u}$ in the lemma statement.
That is, the decay of the Fourier coefficients of $h$ will be nearly as good as those of $\mu_*\rvert_{U}$ given by the first property from \pref{def:plantedproperties}.

Motivated by the above discussion, we will quantify the rate of Fourier decay necessary for our application.
\begin{definition}
	\label{def:rapidlydecaying}
	A function $\epsilon(s_x,s_I)$ is \emph{$b$-rapidly decaying} if
	\begin{enumerate}
		\item $\epsilon(s_x,s_I) = o(1)$ unless $s_x = s_I = 0$.
		\item $\sum_{s_x \geq 1} \epsilon(s_x,1)^c = o(1)$ for any constant $c > 0$
		\item There exists $\nu > 0$ such that for all $s_x \leq d_x$ and $u \geq d_I - 2b$ 
		\[
			2^b\epsilon(s_x,u) \leq \epsilon(s_x,1)^{(1-\nu)}
		\]
	\end{enumerate}
\end{definition}

The final ingredient for the proof of the main result of this section is the following concentration bound based on hypercontractivity.
\begin{theorem}[\cite{OD14} Chapter 9]
	\label{thm:lowdegconcentration}
	Let $f:\{-1,1\}^n \to \R$ have degree at most $k$. Then for any $t \geq \sqrt{2e}^k$ we have
	\[
		\Pr_{x \sim \{-1,1\}^n}\Brac{\abs{f(x)} \geq t \norm{f}_2} \leq \exp\left(-\frac{k}{2e}t^{\frac{2}{k}}\right).
	\]
\end{theorem}

We now prove that averaging $\bar{\mu}_*$ over a CBD distribution on instances produces a function that is non-negative with high probability.
\begin{proposition}
	\label{prop:nonnegwhp}
	 Let $\D$ be a $b$-CBD distribution with parameter $\delta$ and fixed block $U$. Suppose that $\mu_*$ exhibits $\epsilon(s_x,s_I)$-conditional Fourier decay for a $b$-rapidly decaying function $\epsilon(s_x,s_I)$. Let $H(x) = \E_{I \sim \D}\Brac{\bar{\mu}_*\left(x,I\right)}$. Then there is a $\nu>0$ such that
	\[
		\Pr_{x}\Brac{H(x) \leq -o(1)} \leq \sum_{s = 1}^{d_x}\exp\left(-\frac{s}{2e}\epsilon(s,1)^{-\frac{2 - 2\nu}{s}}\right).
	\]
\end{proposition}
\begin{proof}
	We define 
	\begin{align*}
		f(x) &\defeq \Ex[I \sim \D]{L_{d'}\mu_*\rvert_U\left(x,I_{\setcomp{U}}\right)}\\ 
		g(x) &\defeq \Ex[I \sim \D]{h\left(x,I_{\setcomp{U}}\right)} = \Ex[I \sim \D]{(L_{d'}R_U\mu_* - L_{d''}R_U\mu_*) + (R_U\bar{\mu}_* - L_{d'}R_U\bar{\mu}_*)}. 
	\end{align*}
	We have by \pref{lem:restrictcondition} that
	\[
		H(x) = \pi_U(x_U)\cdot f(x) + g(x)
	\]
	Observe that both $L_{d'}\mu_*\rvert_U$ and $h$ do not depend on the coordinates in the fixed block $U$, since these have already been fixed in the definition of each function. 
	This implies that, when sampling from $\D$, the distribution on the input $I_{\setcomp{U}}$ to both $L_{d'}\mu_*\rvert_U$ and $h$ is in fact blockwise-dense.
		Let $f^{=s},g^{=s}$ denote the degree-$s$ part of $f$ and $g$ respectively.
	By the first property of \pref{def:plantedproperties} we have
	\begin{align*}
		\norm{f^{=s}}^2_2 
			&= \sum_{\abs{\alpha} = s} \widehat{f}(\alpha)^2\\
			&= \sum_{\abs{\alpha} = s} \left(\sum_{\abs{\sigma} \leq d_I - \Abs{U}}\widehat{\mu_*\rvert_U}(\alpha,\sigma)\Ex[\D]{\zeta_\sigma(I)}\right)^2\\
			&\leq \sum_{\abs{\alpha} = s} \left(\sum_{\abs{\sigma} \leq d_I - \Abs{U}}\Abs{\widehat{\mu_*\rvert_U}(\alpha,\sigma)\Ex[\D]{\zeta_\sigma(I)}}\right)^2\\
			&\leq \sum_{\abs{\alpha} = s} \binom{n}{s}^{-1}\epsilon(s,1)^2\\
			&= \epsilon(s,1)^2
	\end{align*}
	By \prettyref{thm:lowdegconcentration} we have
	\[
		\Pr_{x}\Brac{\abs{f^{=s}(x)} \geq \tau \epsilon(s,1)} \leq \exp\left(-\frac{s}{2e}\tau^{\frac{2}{s}}\right).
	\]
	Let $c = c(\nu) > 0$ be a sufficiently small constant and set $\tau = \epsilon(s,1)^{-1+c}$. Taking a union bound over $s \geq 1$ yields
	\begin{align*}
		\Pr_{x}\Brac{\Abs{\sum_{s=1}^{d_x}f^{=s}(x)} \geq \sum_{s=1}^{d_x}\epsilon(s,1)^{c}}
			&\leq \sum_{s = 1}^{d_x}\exp\left(-\frac{s}{2e}\epsilon(s,1)^{-\frac{2 - 2c}{s}}\right)\\
	\end{align*}
	Next observe that since $\mu_*\rvert_U$ is a density, the constant term is at least $\Ex{f} \geq 1 - \epsilon(0,1) = 1 - o(1)$. Further, since $\epsilon(s_x,s_I)$ is rapidly decaying, $\sum_{s \geq 1} \epsilon(s,1)^c = o(1)$. Thus we can conclude that
	\[
		\Pr_{x}\Brac{\Abs{f(x)} \leq 0} \leq \sum_{s = 1}^{d_x}\exp\left(-\frac{s}{2e}\epsilon(s,1)^{-\frac{2 - 2c}{s}}\right).
	\]
	
	For $g$ observe that by \pref{lem:hfourierbounds}, and an identical calculation to the one for $f$ yields $\norm{g^{=s}}_2^2 \leq 2^{2\Abs{U}}\epsilon(s,u)^2$ where $u = \max(s,d_I - 2\Abs{U})$. Since $\epsilon(s_x,s_I)$ is rapidly decaying, we have $\norm{g^{=s}}_2^2 \leq \epsilon(s,1)^{2(1-\nu)}$.
	Now applying our above argument for $f$ to $g$ we get that, for $c' > 0$
	\[
		\Pr_{x}\Brac{\Abs{g(x)} \leq -o(1)} \leq \sum_{s = 1}^{d_x}\exp\left(-\frac{s}{2e}\epsilon(s,1)^{-\frac{2 - 2c'}{s}}\right).
	\]
	where the only difference that arises is the fact that $\Ex{g} \geq -o(1)$ since the constant term of $g$ may be negative, but is bounded in magnitude by $2^{\abs{U}}\epsilon(0,u) \leq \epsilon(0,1)^{1 - \nu} = o(1)$. Combining these bounds with the expression $H = \pi_U \cdot f + g$ and the fact that $\pi_U$ is non-negative completes the proof.
%	Plugging this in to \ref{eqn:sumrestrictions} completes the proof.
\end{proof}
\section{Pseudocalibration for CSPs}
\label{sec:csppseudocalibration}
In this section we introduce the planted density for CSPs and prove that this density exhibits conditional Fourier decay whenever the CSP predicate supports a $(t-1)$-wise uniform distribution on satisfying assignments.
 
\subsection{Functions on CSPs}
We now introduce the technical framework that we will use to analyze functions on instances and assignments to CSPs. Recall that $\Enk$ denotes the set of possible scopes $S \in [n]^k$. Clearly $|\Enk| = \frac{n!}{(n-k)!}$ is the number of possible constraint scopes.

We will encode assignments by $x\in \{-1,1\}^n$, sets of constraint scopes by $y \in \{-1,1\}^{|\Enk|}$ and sets of negations by $b\in \{-1,1\}^{|\Enk| \cdot k}$. We will index coordinates of $x$ by indices $i \in [n]$, coordinates of $y$ by scopes $S \in \Enk$ and coordinates of $b$ by pairs $(S,j)$ of scopes $S \in \Enk$ and indices $j \in [k]$. 

The coordinates of $x$ are the assignment to the $n$ variables in the instance. The coordinate $b_{(S,j)}$ of $b$ is the $j$-th entry of the string of negations $b_S$ for the constraint scope $S$. The coordinate $y_S$ is the $\{-1,1\}$ indicator of whether the constraint on scope $S$ is included. We will also use the notation $x_S = (x_{S_1},\dots, x_{S_k})$ and $b_S = (b_{(S,1)},\dots,b_{(S,k)})$. 

Note that every instance $I$ in the support of $\D(p)$ can be described by a pair $(y,b)$ where $y$ indicates which constraint scopes $S$ are included in $I$ and $b$ indicates which negations are applied on each scope, including those not included in $I$. We allow these coordinates of $b$ corresponding to scopes $S$ not included in $I$ to be arbitrary without affecting the instance. For technical reasons, it will be useful to extend the distribution $\D(p)$ to a distribution on pairs $(y,b)$, where the coordinates $b_S$ for scopes $S$ not included in $I$ by $y$ are also sampled uniformly at random.

For a subset $U \subseteq \Enk$ we will write $y_U$ to denote the subset of coordinates of $y$ corresponding to $U$. By extension we will write $b_U$ for the subset of coordinates $b_{(S,j)}$ of the negations $b$ with $S \in U$. Finally, we will use $I_U \defeq (y_U,b_U)$ to denote the instance $I$ restricted to the subset $U$. We will refer to $I_U$ as a restricted instance.

We will use $\U_P$ to denote the $(t-1)$-wise uniform distribution supported on $z \in \{-1,1\}^k$ with $P(z) = 1$. 
We will use $\eta_P$ to denote the density of $\U_P$ with respect to the uniform distribution on $\{-1,1\}^k$.

Using this notation, we can rewrite the $\CSP(P)$ objective function from \prettyref{def:CSP} as
\begin{equation}
\label{eqn:objectiveF}
F(x,y,b) \defeq f_{(y,b)}(x) = \sum_{S\in\Enk} \Ind(y_S = -1)P(b_S \circ x_S)
\end{equation}

\subsection{The Planted Density for CSPs}
We begin with the definition of the planted distribution $\D_*$ as described in \pref{sec:pseudocalibration}.
\begin{definition}
	\label{def:cspplanted}
	An assignment and instance pair $(x,I)$ is sampled from the planted distribution $\D_*$ as follows:
	\begin{itemize}
		\item Sample $x$ uniformly at random from $\{-1,1\}^n$.
		\item Choose to include each constraint scope $S$ independently with probability $p$.
		\item For each scope $S$ chosen, sample $z_S \sim \U_P$ and set $b_S = z_S \circ x_S$.
	\end{itemize}
\end{definition}
The null distribution $\D_\emptyset$ in this case is just the uniform distribution on CSPs $\D(p)$.
As with the distribution $\D(p)$, we can extend $\D_*$ to a distribution on triples $(x,y,b)$ by sampling the coordinates $b_S$ for scopes $S$ not included in $y$ uniformly at random, as these do not affect the satisfiability of the instance. Let $\mu_*(x,y,b)$ be the density of $\D_*$ relative to the distribution $\{-1,1\}^n \times \D(p)$.

Now we claim that
\begin{equation}
\label{eqn:mustarformula}
\mu_*(x,y,b) \propto \prod_{S \in \Enk} \left(\Ind(y_S = -1)\eta_P(b_S \circ x_S) + \Ind(y_S = 1)\right).
\end{equation}
Indeed for any fixed $x,y,b$ we have
\begin{align*}
\Pr_{\substack{(y',b') \sim \D(p) \\ x' \sim \{-1,1\}^n}}[(x',y',b') = (x,y,b)]\mu_*(x,y,b)
&= \left(2^{-(n + k|\Enk|)}\prod_{S : y_S = -1}p \prod_{S : y_S = 1} (1-p)\right) \cdot \mu_*(x,y,b)\\
&\propto 2^{-n}\prod_{S : y_S = -1}  2^{-k} \eta_P(b_S \circ x_S) p \prod_{S : y_S = 1} 2^{-k}(1-p)\\
&= 2^{-n}\prod_{S : y_S = -1} \Pr_{z \sim \U_P}[z = b_S \circ x_S] p \prod_{S : y_S = 1} 2^{-k}(1-p)\\
&= \Pr_{(x',y',b') \sim \D_*}[(x',y',b') = (x,y,b)]
\end{align*}

Next we describe the appropriate Fourier basis for expanding $\mu_*$. Let $\chi_\alpha$ be the Fourier basis for $L^2(\{-1,1\}^n)$. That is, for $\alpha \subseteq [n]$
\[
\chi_\alpha(x) = \prod_{i \in \alpha} x_i.
\]
Let $\phi_\beta$ be the $p$-biased Fourier basis for $L^2(\{-1,1\}^{|\Enk|},\pi_p^{\tensor |\Enk|})$. That is, for $\beta \subseteq \Enk$
\[
\phi_\beta(y) = \prod_{S \in \beta} \phi(y_S), \qquad
\phi(-1) = -\sqrt{\frac{q}{p}}, \qquad
\phi(1) = \sqrt{\frac{p}{q}}
\]
where $q = 1 - p$.
Finally, let $\psi_\gamma$ be the Fourier basis for $L^2(\{-1,1\}^{|\Enk|\cdot k})$. That is, for $\gamma \subseteq \Enk \times [k]$
\[
\psi_\gamma(b) = \prod_{(S,j) \in \gamma} b_{(S,j)}.
\]

We can now write $\mu_*(x,y,b)$ as a function in the tensor product of the spaces spanned by the $\chi_\alpha$, $\phi_\beta$ and $\psi_\gamma$.
That is
\[
	\mu_*(x,y,b) = \sum_{\alpha,\beta,\gamma} \widehat{\mu_*}(\alpha,\beta,\gamma) \chi_\alpha(x)\phi_\beta(y)\psi_\gamma(b)
\]
where the Fourier coefficients above are given by the inversion formula
\[
	\widehat{\mu_*}(\alpha,\beta,\gamma) = \Ex[(x,y,b) \sim \D(p)]{\mu_*(x,y,b)\chi_\alpha(x)\phi_\beta(y)\psi_\gamma(b)} = \Ex[(x,y,b) \sim \D_*]{\chi_\alpha(x)\phi_\beta(y)\psi_\gamma(b)}.
\]

Now letting $M = \abs{\Enk}$ and $\Omega = \{-1,1\} \times \{-1,1\}^{k}$ we have that each instance $I = (y,b)$ is an element of $\Omega^M$.
Further, we can identify $\sigma \in \Abs{\Omega}^M$ with the pair of subsets $\beta \subseteq \Enk$ and $\gamma \subseteq \Enk \times [k]$ as follows.
Each coordinate $\sigma_j$ corresponds to the $j$-th scope $S$ under some ordering of the scopes in $\Enk$. The first bit of $\sigma_j$ indicates whether $S$ is included in $\beta$. The remaining $k$ bits of $\sigma_j$ indicate which pairs $(S,l)$ are included in $\gamma$.
With this correspondence we can define 
\[
	\zeta_\sigma(I) \defeq \phi_\beta(y)\psi_\gamma(b) 
\]
This is indeed a Fourier basis for $L^2(\Omega^M,\D(p))$. Thus, we are in exact setting that is required for \pref{def:pseudocalibration} to make sense.

In particular, for a pair $d = (d_x,d_I) \in \N^2$ the set $A(d)$ corresponds in this case to the set of Fourier coefficients where $x$ has degree $d_x$, at most $d_I$ scopes are included in $\beta$, and the number of scopes whose negations are included in $\gamma$ is also at most $d_I$.
Therefore the pseudo-calibrated density for CSPs $\bar{\mu}_* = L_d \mu_*$ is given by	
\[
	\bar{\mu}_*(x,y,b) = \sum_{(\alpha,\beta,\gamma)\in A(d)} \widehat{\mu_*}(\alpha,\beta,\gamma) \chi_\alpha(x)\phi_\beta(y)\psi_\gamma(b)
\]

\subsection{The Fourier Coefficients of the Planted Density}
In this section we will compute bounds on the Fourier coefficients of the planted density $\mu_*(x,y,b)$.
A key fact that we will exploit is that $\mu_*$ does not directly depend on $x$, but rather only on tuples of the form $b_S \circ x_S = (b_{(S,1)}x_{S_1},\dots, b_{(S,k)}x_{S_k})$. In order to formalize this statement we will need the following definition:
\begin{definition}
	\label{def:derives}
	Let $\gamma \subseteq \Enk \times [k]$ and let $\alpha \subset [n]$. Let $c_i = |\{(S,j) \in \gamma \mid S_j = i\}|$ be the number of appearances of the coordinate $i$ as $S_j$ for some $(S,j) \in \gamma$. Then $\gamma \vdash \alpha$ (in words, $\gamma$ derives $\alpha$) if $c_i$ is odd for every $i \in \alpha$, and $c_i$ is even for every $i \notin \alpha$.
\end{definition}
Recalling that for a pair $(S,j)\in\gamma$ the value $S_j$ is simply an index in $[n]$, the above definition says that every index $i\in \alpha$ must appear an odd number of times as some $S_j$, and every $i \notin \alpha$ must appear an even number of times. To see why this definition is useful consider the following example. Suppose $S_1 = 1$, $S_2 = 2$, $T_1 = 2$ and $T_2 = 3$. Then
\begin{align*}
b_{(S,1)}x_{S_1}b_{(S,2)}x_{S_2} \cdot b_{(T,1)}x_{T_1}b_{(T,2)}x_{T_2} &= b_{(S,1)}x_1b_{(S,2)}x_2 \cdot b_{(T,1)}x_2 b_{(T,2)}x_3\\
&= b_{(S,1)}b_{(S,2)}b_{(T,1)}b_{(T,2)}x_1 x_3\\
&\defeq \psi_\gamma(b)\chi_\alpha(x)
\end{align*}
where the second to last equality used the fact that $x_2^2 = 1$. Thus, when multiplying these two monomials, corresponding to scope $S$ and $T$ respectively, the resulting product has the form $\psi_\gamma(b)\chi_\alpha(x)$ where $\gamma \vdash \alpha$. This is due to the fact that $x_2$ appeared an even number of times and so was eliminated from the product.

We now proceed with the computation of the Fourier coefficient bounds for $\mu_*$.
\begin{lemma}
	\label{lem:derives}
	If $\widehat{\mu_*}(\alpha,\beta,\gamma) \neq 0$ then $\gamma \vdash \alpha$.
\end{lemma}
\begin{proof}
	Let $\eta_P$ be the density of the $(t-1)$-wise independent distribution supported on $P^{-1}(1)$,
	and let $\eta_P(z) = \sum_T \widehat{\eta_P}(T)\prod_{j\in T} z_j$ be its Fourier expansion.
	Recall from \prettyref{eqn:mustarformula}
	\[
		\mu_*(x,y,b) \propto \prod_{S \in \Enk} \Ind(y_S = -1)\eta_P(b_S \circ x_S) + \Ind(y_S = 1).
	\]
	The only terms in the above product depending on either $b$ or $x$ are of the form
	\[
		\eta_P(b_S \circ x_S) = \sum_{T \subseteq k} \widehat{\eta_P}(T) \prod_{j \in T} b_{(S,j)}x_{S_j}.
	\]
	We can use this formula to expand the expression for $\mu_*$ as a sum of products of the variables $b$ and $x$ and the indicators $\Ind(y_S = \pm 1)$.
	Any term depending on $b$ or $x$ will then be a product, over \emph{distinct} $S$, of terms of the form
	\[
		\prod_S \Ind(y_S = \pm 1)\widehat{\eta_P}(T_S)\prod_{j \in T_S} b_{(S,j)}x_{S_j}
	\]
	where the $T_S \subseteq [k]$ are a sequence of subsets depending on $S$.
	In any such product over distinct $S$ all the variables $b_{(S,j)}$ for $j \in T_S$ are distinct.
	So letting $\gamma = \cup_S \{(S,j) \mid j \in T_S\}$, we have
	\[
		\prod_S \Ind(y_S = \pm 1)\widehat{\eta_P}(T_S)\prod_{j \in T_S} b_{(S,j)}x_{S_j} = \psi_\gamma(b)\prod_S \Ind(y_S = \pm 1)\widehat{\eta_P}(T_S)\prod_{j \in T_S} x_{S_j}.
	\]
	Furthermore, every $x_i$ appears exactly $c_i = |\{(S,j) \in \gamma \mid S_j = i\}|$ times in the product.
	If $c_i$ is even $x_i^{c_i} = 1$, and if $c_i$ is odd $x_i^{c_i} = x_i$. Thus,
	\[
		\psi_\gamma(b)\prod_S \Ind(y_S = \pm 1)\widehat{\eta_P}(T_S)\prod_{j \in T_S} x_{S_j} = \psi_\gamma(b)\chi_\alpha(x)\prod_S \widehat{\eta_P}(T_S)\Ind(y_S = \pm 1)
	\]
	where $\gamma \vdash \alpha$.
	Now if we expand the indicator functions $\Ind(y_S = \pm 1)$ in the $\phi_\beta$ basis, we conclude that $\mu_*(x,y,b)$ can be written as a linear combination of terms $\chi_\alpha(x)\phi_\beta(y)\psi_\gamma(b)$ where all such terms appearing with a non-zero coefficient satisfy $\gamma \vdash \alpha$.
	There may be additional cancellations between terms in the final Fourier expansion of $\mu_*$ as some of the basis functions with non-zero coefficients could appear multiple times.
	However, we can still conclude that any non-zero coefficients $\widehat{\mu_*}(\alpha,\beta,\gamma)$ must have $\gamma \vdash \alpha$.

\end{proof}

Two additional definitions are required in order to state the formula for the Fourier coefficients. First, we need notation for the set of scopes contained in $\gamma \subseteq \Enk \times [k]$.
\begin{definition}
	Let $\gamma \subseteq \Enk \times [k]$. Then $\bar{\gamma} \defeq \{S \mid (S,j) \in \gamma \text{ for some } j \}$ is the set of scopes $S$ present in $\gamma$.
\end{definition}
Second, we need notation for the minimum number of coordinates contained in any scope $S$ present in $\gamma$.
\begin{definition}
	Let $\gamma \subseteq \Enk \times [k]$. Then $r(\gamma) \defeq \min_{S \in \bar{\gamma}} \abs{\{j \mid (S,j) \in \gamma\}}$ is the \emph{minimum arity} of $\gamma$.
\end{definition}

Now using \prettyref{lem:derives} we can compute bounds on the Fourier coefficients of $\mu_*$.
\begin{lemma}
	\label{lem:fourierformula}
	If $\gamma \vdash \alpha$, $r(\gamma) \geq t$ and $\beta \subseteq \bar{\gamma}$ then
	\[
		\abs{\widehat{\mu_*}(\alpha,\beta,\gamma)} \leq  \sqrt{pq}^{\abs{\bar{\gamma}\cap\beta}}  p^{\abs{\bar{\gamma}\setminus\beta}}
	\]
	otherwise $\widehat{\mu_*}(\alpha,\beta,\gamma) = 0.$
\end{lemma}
\begin{proof}
	The Fourier coefficients are given by
	\[
		\widehat{\mu_*}(\alpha,\beta,\gamma) = \Ex[(x,y,b) \sim \D(p)]{\mu_*(x,y,b)\chi_\alpha(x)\phi_\beta(y)\psi_\gamma(b)} = \Ex[(x,y,b) \sim \D_*]{\chi_\alpha(x)\phi_\beta(y)\psi_\gamma(b)}.
	\]
	By \prettyref{lem:derives} any non-zero Fourier coefficients must have $\gamma \vdash \alpha$. In this case, letting $c_i = |\{(S,j) \in \gamma \mid S_j = i\}|$ as in \prettyref{def:derives}, we have
	\begin{align*}
		\chi_\alpha(x)\phi_\beta(y)\psi_\gamma(b)
			&= \phi_\beta(y)\prod_{(S,j) \in \gamma} b_{(S,j)}\prod_{i\in\alpha}x_i\\
			&= \phi_\beta(y)\prod_{(S,j) \in \gamma} b_{(S,j)}\prod_{i\in [n]} x_i^{c_i}\\
		 	&= \phi_\beta(y)\prod_{(S,j) \in \gamma} b_{(S,j)}x_{S_j}
	\end{align*}
	where the second equality used the fact that $\gamma \vdash \alpha$ implies that $x_i^{c_i} = x_i$ for $i \in \alpha$ and $x_i^{c_i} = 1$ for $i \notin \alpha$.
	Expanding $\phi_\beta$ as a product over scopes $S \in \Enk$ and grouping terms yields
	\[
		\chi_\alpha(x)\phi_\beta(y)\psi_\gamma(b) = \left(\prod_{S \in \beta\setminus\bar{\gamma}} \phi(y_S)\right)
		\left(\prod_{S \in \beta\cap\bar{\gamma}} \phi(y_S) \prod_{j : (S,j) \in \gamma} b_{(S,j)}x_{S_j}\right)
		\left(\prod_{S \in \bar{\gamma}\setminus\beta}\prod_{j : (S,j) \in \gamma} b_{(S,j)}x_{S_j}\right).
	\]
	Now we claim that, for each $S\in \Enk$, the term corresponding to $S$ in the above product is independent of all the other terms under the distribution $\D_*$. This is clearly true for the $S \in \beta\setminus\bar{\gamma}$ terms, since under $\D_*$ the variables $y_S$ depend only on the variables $x_S$ and negations $b_S$ on scope $S$, but none of those variables appear in the above product.
	For the remaining terms, note that the distribution of $b_S \circ x_S$ is independent of all other variables and constraints except $y_S$.
	This is because under $\D_*$ both $x$ and $y$ are sampled independently, and then each $b_S$ is sampled independently so that if $y_S = -1$ the distribution of $b_S \circ x_S$ is equal to the $(t-1)$-wise uniform distribution supported on satisfying assignments to $P$, and if $y_S = 1$ then $b_S \circ x_S$ is uniform and independent of everything else. Therefore, taking the expectation of $\chi_\alpha(x)\phi_\beta(y)\psi_\gamma(b)$ over $\D_*$ yields the product of expectations
	\begin{align*}
		&\Ex[(x,y,b) \sim \D_*]{\chi_\alpha(x)\phi_\beta(y)\psi_\gamma(b)}\\
			&= {\left(\prod_{S \in \beta\setminus\bar{\gamma}} \Ex{\phi(y_S)}\right)
				\left(\prod_{S \in \beta\cap\bar{\gamma}} \Ex{\phi(y_S) \prod_{j : (S,j) \in \gamma} b_{(S,j)}x_{S_j}}\right)
				\left(\prod_{S \in \bar{\gamma}\setminus\beta}\Ex{\prod_{j : (S,j) \in \gamma} b_{(S,j)}x_{S_j}}\right)}.
	\end{align*}
	First, note that $\Ex[\D_*]{\phi(y_S)} = \Ex[\D(p)]{\phi(y_S)}= 0$, and so the above product is zero whenever $\beta\setminus\bar{\gamma} \neq \emptyset$. Thus all non-zero Fourier coefficients $\widehat{\mu_*}(\alpha,\beta,\gamma)$ must have $\beta \subseteq \bar{\gamma}$.

	Second, suppose there exists some $S^* \in \bar{\gamma}$ such that $\abs{\{j \mid (S^*,j)\in\gamma\}} < t$.
	Then the distribution of the coordinates $\{b_{(S^*,j)}x_{S^*_j}\}_{(S^*,j)\in\gamma}$ is uniform and independent of everything else. This is because $b_S \circ x_S$ is sampled independently from a $(t-1)$-wise uniform distribution if $y_S = -1$ and from the uniform distribution if $y_S = 1$. Thus, in both cases the aforementioned coordinates are uniform and independent of all other variables. This implies that
	\[
		\Ex{\prod_{j : (S^*,j) \in \gamma} b_{(S^*,j)}x_{S^*_j}} = 0 = \Ex{\phi(y_{S^*})\prod_{j : (S^*,j) \in \gamma} b_{(S^*,j)}x_{S^*_j}} .
	\]
	So we conclude that whenever such an $S^*$ exists, the Fourier coefficient is zero. Taking the contrapositive, we have that all nonzero Fourier coefficients must have $r(\gamma) \geq t$.

	Next, for $S\in \beta \cap \bar{\gamma}$
	\begin{align*}
		\Abs{\Ex[\D_*]{\phi(y_S) \prod_{j : (S,j) \in \gamma} b_{(S,j)}x_{S_j}}}
			&= \Abs{p \cdot \phi(-1) \Ex{\prod_{j : (S,j) \in \gamma} b_{(S,j)}x_{S_j} \mid y_S = -1} + q \cdot \phi(1)\cdot 0 }\\
			&\leq \Abs{p \cdot \phi(-1)} = \sqrt{pq}
	\end{align*}
	where in the first equality we have used the fact that $b_S \circ x_S$ is uniformly random conditioned on $y_S = 1$, and the inequality follows from the fact that the variables $b$ and $x$ are bounded by one.
	Finally, for $S \in \bar{\gamma}\setminus\beta$, we have by the same argument
	\begin{align*}
		\Abs{\Ex[\D_*]{\prod_{j : (S,j) \in \gamma} b_{(S,j)}x_{S_j}}}
			= \Abs{p \cdot \Ex{\prod_{j : (S,j) \in \gamma} b_{(S,j)}x_{S_j} \mid y_S = -1} + q \cdot 0 } \leq p.
	\end{align*}
	Plugging these bounds into the original product of expectations completes the proof.
\end{proof}

\subsection{The Planted Density for $\maxkxor$ and $\maxksat$}
At this point it is instructive to see what the planted density looks like for a concrete example. For the $\maxkxor$ problem the predicate is simply the sum of the input bits modulo two. Thus, $P(x) = \frac{1}{2} - \frac{1}{2}\prod_i x_i$. The uniform distribution on all satisfying assignments to $P$ is $k-1$-wise independent. Also note that the planted density for $\maxksat$ is exactly the same as that for $\maxkxor$. This is because the $k-1$-wise independent distribution supported on satisfying assignments to the $\maxksat$ predicate can be taken to be the uniform distribution on assignments satisfying the $\maxkxor$ predicate on the same bits. Combining these facts, and using the above analysis we can precisely compute the Fourier expansion of $\bar{\mu}_*$ for $\maxkxor$ (and $\maxksat$).
\[
	\mu_*(x,y,b) = \sum_{\alpha}\chi_{\alpha}(x)\sum_{\substack{\gamma \vdash \alpha\\ r(\gamma) \geq k}}\chi_\gamma(b)\sum_{\beta \subseteq{\bar{\gamma}}}(-\sqrt{pq})^{\beta \cap \bar{\gamma}}p^{\bar{\gamma}\setminus\beta}\phi_\beta(S)
\]

Since each scope has size $k$, the fact that $r(\gamma)\geq k$ implies that for each scope $S$, either the full product of the $k$ bits $b_S$ appears in $\chi_\gamma$ or none of them do. Thus, all that matters in the above Fourier expansion is if the parity of $b_S$ is odd or even. This makes sense, because for the $\maxkxor$ predicate all that is relevant about the negations is their parity. Let us use $c_S$ to denote this parity and let $c_\gamma \defeq \prod_{S \in \bar{\gamma}} c_S$.
Now, observing that $\Ind(y_S = -1) = p - \sqrt{pq}\phi(y_S)$ we can rewrite the above expression as
\[
	\mu_*(x,y,b) = \sum_{\alpha}\chi_{\alpha}(x)\sum_{\substack{\gamma \vdash \alpha\\ r(\gamma) = k}}\Ind(y_{\bar{\gamma}} = -\ovec)c_\gamma.
\]
Now imagine fixing $y$ and $b$ to some value i.e. fixing some instance of $\maxkxor$. In this case, the above expansion of $\mu_*$ simplifies to a function of $x$. For each parity $\chi_\alpha(x)$ we have a sum over derivations of $\alpha$ by the constraints included in the instance $y$. The coefficients of the sum are $\pm 1$ depending on whether the negations on a given constraint have odd or even parity.
Here the derivations of $\alpha$ by some set of constraints correspond exactly to linear combinations over $\F_2$ of the set of constraints when thought of as vectors in $\F_2$.

\subsection{Counting Non-zero Fourier Coefficients}
Now that we have computed bounds for the Fourier coefficients of the planted density $\mu_*$ we can use them to prove that $\mu_*$ has the conditional Fourier decay properties from \pref{def:plantedproperties}.
The first step is to estimate the number of nonzero Fourier coefficients $\widehat{\mu_*}(\alpha,\beta,\gamma)$ for each fixed $\alpha$.
\begin{lemma}
	\label{lem:fouriercount}
	Fix $\alpha \subseteq [n]$ and let $l \leq \frac{cn}{k}$ for a sufficiently small constant $c > 0$. Let $N_l(\alpha)$ denote the number of Fourier coefficients with $\widehat{\mu_*}(\alpha,\beta,\gamma) \neq 0$ and $\abs{\bar{\gamma}} = l$. Then
	\[
		N_l(\alpha) \leq C^l n^{kl - \frac{tl + \abs{\alpha}}{2}} l^{\frac{tl + \abs{\alpha}}{2} - l}
	\]
	where $C$ is a constant depending only on $k$ and $t$.
\end{lemma}
\begin{proof}
	By \prettyref{lem:fourierformula} all nonzero Fourier coefficients must satisfy $\gamma \vdash \alpha$, $r(\gamma) \geq t$ and $\beta \subseteq \bar{\gamma}$. For fixed $\alpha$ we first count the number of $\gamma$ with $\abs{\bar{\gamma}} = l$, $\gamma \vdash \alpha$, and $r(\gamma) \geq t$.
	To do so, note that each scope in $\bar{\gamma}$ is a sequence of $k$ indices from $[n]$. Thus, we can think of choosing each scope $S \in \bar{\gamma}$ by filling in $k$ empty slots with indices from $[n]$.

	Each possible $\gamma$ can then be selected as follows: First choose $l$ subsets $T_1,\dots, T_l \subseteq [k]$ to be the subsets of slots which are to be included in $\gamma$ from each of the $l$ scopes in $\bar{\gamma}$. Let $s_i = \abs{T_i}$ and let $s = \sum_i s_i$. Next, assign indices in $[n]$ to each of the $s$ slots in $T_1,\dots, T_l$ so that $\gamma \vdash \alpha$. Finally, assign indices arbitrarily to the remaining $kl - s$ slots in each scope in $\bar{\gamma}$.

	Now we count the number of ways to choose $\gamma$ by the above process.
	First, we fix a choice for the subsets $T_i$.
	For $\gamma \vdash \alpha$ to hold, some subset $A$ of the $s$ slots selected by the $T_i$ must be assigned the indices from $\alpha$, since each such index appears an odd number of times in $\gamma$. There are $\binom{s}{\abs{\alpha}}$ ways to select the subset $A$ and $\abs{\alpha}!$ ways to assign the indices from $\alpha$ to $A$.

	Further, the $s - \abs{\alpha}$ slots outside of $A$ must be assigned indices that come in matching pairs, since each such index must appear an even number of times.
	In this case the total number $d$ of distinct indices appearing in slots outside of $A$ is at most $\frac{s - \abs{\alpha}}{2}$. This follows from the fact that there are $s - \abs{\alpha}$ slots outside of $A$ and each index must appear at least twice. There are $\binom{n}{d}$ choices for the $d$ indices and at most $d^{s-\abs{\alpha}}$ ways to assign these indices to the slots selected by the $T_i$ that are not in $A$.

	Finally, there are $n^{kl - s}$ ways to arbitrarily assign the remaining $kl - s$ slots of the scopes in $\bar{\gamma}$. Putting this all together we have that the total number of ways to assign indices to the $s$ slots is
	\begin{align*}
		 \binom{s}{\abs{\alpha}}\abs{\alpha}!\sum_{d = 1}^{\frac{s - \abs{\alpha}}{2}} \binom{n}{d} d^{s - \abs{\alpha}} n^{kl - s}
			&\leq \binom{s}{\abs{\alpha}}\abs{\alpha}!\left(\frac{s - \abs{\alpha}}{2}\right) \binom{n}{\frac{s - \abs{\alpha}}{2}} \left(\frac{s - \abs{\alpha}}{2}\right)^{s - \abs{\alpha}}n^{kl - s}\\
			&\leq (es)^{\abs{\alpha}}\frac{s}{2} \left(\frac{en}{\frac{(s - \abs{\alpha})}{2}}\right)^{\frac{s - \abs{\alpha}}{2}} \left(\frac{s - \abs{\alpha}}{2}\right)^{s - \abs{\alpha}}n^{kl - s}\\
			&\leq \frac{s}{2}(en)^{kl - \frac{s + \abs{\alpha}}{2}}(es)^{\frac{s + \abs{\alpha}}{2}}
	\end{align*}
	where the second to last inequality uses the bounds $\binom{n}{k} \leq \left(\frac{en}{k}\right)^k$ and $n! \leq n^n$.

	Now note that if we sum the above bound over all choices for the subsets $T_1,\dots T_l$ each $\gamma$ will be counted $l!$ times, once for each ordering of the $l$ scopes comprising $\gamma$. Thus the total number of $\gamma$ where $\gamma \vdash \alpha$ and $r(\gamma) \geq t$ is at most
	\begin{align*}
		\abs{\{\gamma \mid \gamma \vdash \alpha, r(\gamma) \geq t\}}
			&\leq \frac{1}{l!}\sum_{T_1,\dots,T_l} \frac{s}{2}(en)^{kl - \frac{s + \abs{\alpha}}{2}}(es)^{\frac{s + \abs{\alpha}}{2}}\\
			&\leq \frac{1}{l!}\sum_{T_1,\dots,T_l} \frac{kl}{2}(en)^{kl - \frac{s + \abs{\alpha}}{2}}(ekl)^{\frac{s + \abs{\alpha}}{2}}
	\end{align*}
	where the last line uses that $s \leq kl$. Since $r(\gamma) \geq t$ we have $s \geq tl$. Further $kl < n$, and so the expression inside the sum is maximized when $s$ is minimized i.e. when $s = tl$.
	Plugging this into the above inequality we have
	\begin{align*}
		\abs{\{\gamma \mid \gamma \vdash \alpha, r(\gamma) \geq t\}}
			&\leq \frac{1}{l!}\sum_{s_1,\dots,s_l} \frac{kl}{2}(en)^{kl - \frac{tl + \abs{\alpha}}{2}}(ekl)^{\frac{tl + \abs{\alpha}}{2}}\\
			&\leq \frac{k^l}{l!}\frac{kl}{2}(en)^{kl - \frac{tl + \abs{\alpha}}{2}}(ekl)^{\frac{tl + \abs{\alpha}}{2}}\\
			&\leq l^{-l}2^l(ek)^{l}\frac{k}{2}(en)^{kl - \frac{tl + \abs{\alpha}}{2}}(ekl)^{\frac{tl + \abs{\alpha}}{2}}\\
			&\leq C^l n^{kl - \frac{tl + \abs{\alpha}}{2}} l^{\frac{tl + \abs{\alpha}}{2} - l}
	\end{align*}
	where in the last inequality $C$ is a constant depending only on $k$ and $t$.
	Finally, since $\beta \subseteq \bar{\gamma}$ there are $2^l$ possible values of $\beta$ for each $\gamma$ satisfying $\gamma \vdash \alpha$ and $r(\gamma) \geq t$. Combining this with the above bound completes the proof.
\end{proof}

A key combinatorial quantity in the remainder of the analysis is the sum over non-zero Fourier coefficients weighted by $p^{\abs{\bar{\gamma}}}$. Our next step is to compute bounds on this quantity.

\begin{lemma}
	\label{lem:levell}
	Fix $\alpha \subseteq [n]$. For $\Delta \geq 1$, let $p = \frac{\Delta n}{\abs{\Enk}}$ and let $\ceil{\frac{\alpha}{k}} \leq s \leq l$. Let $N_r(\alpha)$ be the number of Fourier coefficients with $\widehat{\mu_*}(\alpha,\beta,\gamma) \neq 0$ and $\abs{\bar{\gamma}} = r$.
	Let $\rho>0$ be a sufficiently small constant depending only on $k$. For any $l \leq \rho n\Delta^{-\frac{2}{t-2}}$
	\[
		\sum_{r = s}^l p^rN_r(\alpha) 
		\leq (C\Delta)^s\left(\frac{s}{n}\right)^{\frac{t-2}{2}s + \frac{\abs{\alpha}}{2}}
	\]
	for some constant $C$ depending only on $k$.
\end{lemma}
\begin{proof}
	By \prettyref{lem:fouriercount} we have
	\begin{align*}
		\sum_{r = s}^l p^rN_r(\alpha)
		&\leq \sum_{r=s}^l p^{r} \cdot C^r n^{kr - \frac{tr + \abs{\alpha}}{2}} r^{\frac{tr + \abs{\alpha}}{2} - r}\\
		&\leq \sum_{r=s}^l \Delta^r \cdot C^r n^{- \frac{tr + \abs{\alpha}}{2} + r} r^{\frac{tr + \abs{\alpha}}{2} - r}\\
		&\leq \sum_{r=s}^l (C\Delta)^r \left(\frac{r}{n}\right)^{\frac{t - 2}{2}r + \frac{\abs{\alpha}}{2}}\\
		&\leq (C\Delta)^s\left(\frac{s}{n}\right)^{\frac{t-2}{2}s + \frac{\abs{\alpha}}{2}}\sum_{r = s}^l (C\Delta)^{r - s}\left(\frac{r}{n}\right)^{\frac{t-2}{2}r + \frac{\abs{\alpha}}{2}}
		\left(\frac{n}{s}\right)^{\frac{t-2}{2}s + \frac{\abs{\alpha}}{2}}\\
		&= (C\Delta)^s\left(\frac{s}{n}\right)^{\frac{t-2}{2}s + \frac{\abs{\alpha}}{2}}\sum_{r = s}^l (C\Delta)^{r - s}\left(\frac{r}{n}\right)^{\frac{t-2}{2}(r-s)}
		\left(\frac{r}{s}\right)^{\frac{t-2}{2}s + \frac{\abs{\alpha}}{2}}.
	\end{align*}
	We will now show that the sum in the above equation is bounded by a constant. Since $r \leq l = \rho n\Delta^{-\frac{2}{t-2}}$ and $\abs{\alpha}\leq ks$ we have
	\begin{align*}
		\sum_{r = s}^l (C\Delta)^{r - s}\left(\frac{r}{n}\right)^{\frac{t-2}{2}(r-s)}
		\left(\frac{r}{s}\right)^{\frac{t-2}{2}s + \frac{\abs{\alpha}}{2}}
		&\leq \sum_{r = s}^l C^{r - s}\rho^{\frac{t-2}{2}(r-s)}
		\left(\frac{r}{s}\right)^{\frac{t-2}{2}s + \frac{ks}{2}}
	\end{align*}

	Using that $t \geq 3$ and setting $\rho < (e^{3k}C)^{-2}$ we next obtain
	\begin{align*}
		\sum_{r = s}^l \exp{(-3k(r-s))}
		\left(\frac{r}{s}\right)^{\frac{t-2}{2}s + \frac{ks}{2}}
			&\leq \sum_{r = s}^l \exp\left(-3k(r-s) + ks\log \left(\frac{r}{s} \right)\right)\\
			&= 	\sum_{r = s}^l \exp\left(-3k(r-s) + ks\log \left(1 + \frac{r - s}{s} \right)\right)\\
			&\leq \sum_{r = s}^l \exp\left(-3k(r-s) + ks \left(\frac{r-s}{s} \right)\right)\\
			&= \sum_{r = s}^l \exp\left(-2k(r-s)\right) \leq 1 + \frac{1}{\exp(-2k)}
	\end{align*}

	Putting it all together we have shown
	\begin{align*}
	\sum_{r = s}^l p^rN_r(\alpha)
		&\leq (C\Delta)^s\left(\frac{s}{n}\right)^{\frac{t-2}{2}s + \frac{\abs{\alpha}}{2}}
	\end{align*}
\end{proof}

Next for a fixed $\alpha$ we compute the $L^2$-norm of the part of $\mu_*$ corresponding to Fourier coefficients with $\abs{\bar{\gamma}} \leq l$.

\begin{lemma}
	\label{lem:l2norm}
	Fix $\alpha \subseteq [n]$. For $\Delta \geq 1$, let $p = \frac{\Delta n}{\abs{\Enk}}$ and let $s = \ceil*{\frac{\abs{\alpha}}{k}}$.
	Let $\rho>0$ be a sufficiently small constant depending only on $k$. For any $l \leq \rho n\Delta^{-\frac{2}{t-2}}$
	\[
		\sum_{\beta, \gamma: \abs{\bar{\gamma}} \leq l} \widehat{\mu_*}^2(\alpha,\beta,\gamma)
		\leq (C\Delta)^s\binom{n}{s}^{-\frac{t-2}{2}}\binom{n}{\abs{\alpha}}^{-\frac{1}{2}}
	\]
	for some constant $C$ depending only on $k$.
\end{lemma}
\begin{proof}
	By \prettyref{lem:fourierformula} the magnitude of every non-zero Fourier coefficient of $\mu_*$ is bounded by
	\begin{equation*}
		\abs{\widehat{\mu_*}(\alpha,\beta,\gamma)}
			\leq  \sqrt{pq}^{\abs{\bar{\gamma}\cap\beta}}  p^{\abs{\bar{\gamma}\setminus\beta}}
			\leq p^{\frac{\abs{\bar{\gamma}}}{2}}.
	\end{equation*}
	Let $s = \ceil*{\frac{\abs{\alpha}}{k}}$. Combining the above Fourier coefficient bounds with \prettyref{lem:fouriercount} we have
	\begin{align*}
		\sum_{\beta, \gamma: \abs{\bar{\gamma}} \leq l} \widehat{\mu_*}^2(\alpha,\beta,\gamma)
		&\leq \sum_{r=s}^l p^{r} \cdot N_r(\alpha)\\
		&\leq (C\Delta)^s\left(\frac{s}{n}\right)^{\frac{t-2}{2}s + \frac{\abs{\alpha}}{2}}
	\end{align*}
	where the final inequality is \pref{lem:levell}. Now using the fact $s \leq \abs{\alpha}$ and absorbing some constant factors into $C$ we obtain
	\[
		\sum_{\beta, \gamma: \abs{\bar{\gamma}} \leq l} \widehat{\mu_*}^2(\alpha,\beta,\gamma)
		\leq (C\Delta)^s\left(\frac{s}{n}\right)^{\frac{t-2}{2}s} \left(\frac{\abs{\alpha}}{n}\right)^{ \frac{\abs{\alpha}}{2}}
		\leq (C\Delta)^s\binom{n}{s}^{-\frac{t-2}{2}}\binom{n}{\abs{\alpha}}^{-\frac{1}{2}}.
	\]
\end{proof}

\subsection{Conditional Fourier Decay for CSPs}
 In this section we will establish that the pseudo-calibrated density $\bar{\mu}_*$ exhibits conditional Fourier decay.
Recall that for a fixed block $U\subseteq \Enk$ we defined the restricted instance $I_U = (y_U,b_U)$ where $y_U$ is the subset of coordinates of $y$ restricted to the scopes in $U$ and $b_U$ is the subset of coordinates $b_{(S,j)}$ for all $S\in U$ and all $j\in[k]$. We will let $V(U)$ denote the coordinates of $x$ appearing in the scopes in $U$. By abuse of notation we will also write $x_V \defeq x_{V(U)}$.

Using this notation, when we fix the block $U$ to $I_U = (y_U,b_U)$, the conditional density $\mu_*\rvert_U$ from \pref{def:condtionplanted} is given by
\[
	\mu_*\rvert_U(x,y_{\setcomp{U}},b_{\setcomp{U}}) = \frac{\Pr_{\D_*}[x_{\setcomp{V}},y_{\setcomp{U}},b_{\setcomp{U}} \mid x_V, y_U,b_U]}{\Pr_{\D(p)}[x_{\setcomp{V}},y_{\setcomp{U}},b_{\setcomp{U}}]}.
\]

\begin{observation}
\label{obs:mustarcondformula}
	Let $U \subseteq \Enk$. Then we have
	\[
		\mu_*\rvert_U(x,y_{\setcomp{U}},b_{\setcomp{U}}) \propto \prod_{S \notin U} \left(\Ind(y_S = -1)\eta_P(b_S \circ x_S) + \Ind(y_S = 1)\right).
	\]
\end{observation}
\begin{proof}
Observe that since $x_V$ consists of all variables present in the scopes $U$, the distribution of $x_V,I_U$ is independent of everything outside of $U$. Thus we can apply \pref{eqn:mustarformula} to only the scopes in $U$ to obtain.
\[
	\pi_U(x_V,I_U) \propto \prod_{S \in U} \left(\Ind(y_S = -1)\eta_P(b_S \circ x_S) + \Ind(y_S = 1)\right)
\]
By \pref{lem:restrictcondition} we have $R_U\mu_*(x,y_{\setcomp{U}},b_{\setcomp{U}}) = \pi_U(x_V,I_U)\mu_*\rvert_U(x,y_{\setcomp{U}},b_{\setcomp{U}})$. Now applying \pref{eqn:mustarformula} to $R_U\mu_*(x,y_{\setcomp{U}},b_{\setcomp{U}}) = \mu_*(x,y,b)$ and dividing through by the formula for $\pi_U$ completes the proof. 
\end{proof}

%We now prove bounds on $\pi_U$ and $\mu_*\rvert_U$. Recall that $\abs{I}$ denotes the number of constraints in an instance $I$.
%\Jnote{Unclear that this lemma is still necessary}
%\begin{lemma}
%	Let $U \subseteq \Enk$. Then $\pi_U$ is non-negative and bounded by
%	\[
%		\pi_U(x_V,I_U) \leq 2^{k\abs{I_U}}%\Ind\left(F(x_V,I_U) = \abs{I_U}\right).
%	\]
%\end{lemma}
%\begin{proof}
%	Clearly $\pi_U$ is non-negative.
%	Let $u$ be the number of coordinates of $x_V$.
%	Note that since $x_V$ consists of all variables present in the scopes $U$, the distribution of $x_V,I_U$ is independent of everything outside of $U$. Thus
%	\begin{align*}
%		\Pr_{\D_*}[x_V,I_U]
%			&= 2^{-u}p^{\abs{I}}(1-p)^{\abs{U} - \abs{I}}2^{-k(\abs{U} - \abs{I_U})}\prod_{S \in I_U}\Pr_{\U_P}[b_S \circ x_S]\\
%			&\leq 2^{-u}p^{\abs{I}}(1-p)^{\abs{U} - \abs{I}}2^{-k(\abs{U} - \abs{I_U})}
%	\end{align*}
%	Observing that
%	\[
%		\Pr_{D(p)}[x_U,I_U] = 2^{-u}p^{\abs{I}}(1-p)^{\abs{U} - \abs{I}}2^{-k\abs{U}}
%	\]
%	completes the proof.
%\end{proof}

Next we prove bounds on the Fourier coefficients of $\mu_*\rvert_U$ by relating them to the original Fourier coefficients of $\mu_*$. The key fact we exploit is that, after conditioning on all variables $x_V$ in the scopes $U$, the distribution of $x,y,b$ outside of $U$ is independent of $I_U$.
\begin{lemma}
	\label{lem:conditionalfourier}
	Let $U \subseteq \Enk$ and let $I_U$ and $x_V$ be the corresponding restricted instance and assignment. Let $\alpha \subseteq [n]$. If $\gamma \vdash \alpha$, $\bar{\gamma} \subseteq \Enk\setminus U$, $r(\gamma) \geq t$ and $\beta \subseteq \bar{\gamma}$ then
	\[
		\widehat{\mu_*\rvert_U}(\alpha,\beta,\gamma) \leq \sqrt{pq}^{\abs{\bar{\gamma}\cap\beta}}p^{\bar{\gamma}\setminus\beta}.
	\]
	Otherwise, $\widehat{\mu_*\rvert_U}(\alpha,\beta,\gamma) = 0$.
\end{lemma}
\begin{proof}
	The proof is by modifying the proofs in \pref{lem:derives} and \pref{lem:fourierformula} for the original planted density.
	By \pref{obs:mustarcondformula} we have that
	\[
		\mu_*\rvert_U(x,y_{\setcomp{U}},b_{\setcomp{U}}) \propto \prod_{S \notin U} \left(\Ind(y_S = -1)\eta_P(b_S \circ x_S) + \Ind(y_S = 1)\right).
	\]
	Thus the argument from the proof of \pref{lem:derives} still holds and all non-zero Fourier coefficients must have $\bar{\gamma} \vdash \alpha$.
	Further, since the above product is only over $b$ with coordinates outside of $U$, we can write its Fourier expansion entirely in terms of basis functions $\psi_\gamma$ with both $\bar{\gamma} \subseteq \Enk\setminus U$ and $\beta \subseteq \Enk\setminus U$.

	Next, since $\mu_*\rvert_U$ is a density relative to $\D(p)\rvert U$ we have
	\begin{align*}
		\Ex[\D(p)]{\mu_*\rvert_U \chi_\alpha \phi_\beta \psi_\gamma} 
			&= \Ex[x_V, I_U \sim \D(p)]{\Ex[\D(p)\rvert_U]{\mu_*\rvert_U \chi_\alpha \phi_\beta \psi_\gamma}}\\
			&= \Ex[x_V, I_U \sim \D(p)]{\Ex[\D_*\rvert_U]{\chi_{\alpha}\phi_\beta\psi_\gamma}}
	\end{align*}
	Thus if we can get a bound on the inner expectation that holds for all fixings of $x_V, I_U$ then we will have an overall bound on the Fourier coefficients.
	Therefore, we proceed by carrying out the argument of \pref{lem:fourierformula} where we compute the expectation of each basis function
	\[
	\Ex[\D_*\rvert_U]{\chi_{\alpha}(x)\phi_\beta(y)\psi_\gamma(b)}.
	\]
	As before, this expectation splits into a product of expectations corresponding to each scope $S$. To see why, note that since $\gamma \vdash \alpha$, the term corresponding to each scope $S$ has dependence on $x$ and $b$ of the form
	\[
		\prod_{j : (S,j) \in \gamma} b_{(S,j)}x_{S_j}.
	\]
	Therefore, even though the variables $x_V$ are fixed due to the conditioning of $\D_*\rvert_U$, the product above depends only on $y_S$ and is independent of everything else. 
	
	Next, any term where $\beta\setminus\bar{\gamma} \neq \emptyset$ must be zero as we still have that $\Ex[\D_*\rvert_U]{\phi(y_S)} = \Ex[\D(p)]{\phi(y_S)} = 0$. Furthermore, $t$-wise independence of $b_S \circ x_S$ still implies that the expectation corresponding to any scope $S$ with $\Abs{\left\{j \mid (S,j) \in \gamma\right\} } < t$ must be zero. So $r(\gamma) \geq t$ for all non-zero Fourier coefficients.

	Finally, the magnitude of the expectation of each term corresponding to a remaining scope $S$ has the same upper bound as in \pref{lem:fourierformula}, since all we used in that part of the proof was that each scope is included with probability $p$ and that the coordinates of $x$ and $b$ are bounded by one. This completes the proof.
\end{proof}

We are now ready to prove that the planted density for CSPs exhibits conditional Fourier decay.
\begin{proposition}
\label{prop:cspconditionalfourierdecay}
Fix $\delta = O(\log n)$ in \pref{def:plantedproperties}. 
Let $1 < \Delta < n^{\frac{t-2}{2} - \epsilon} $ and let $p = \frac{\Delta n}{\abs{\Enk}}$.
If $s = \max\left\{\ceil*{\frac{s_x}{k}} , s_I\right\}$ then $\mu_*$ exhibits $\epsilon(s_x,s_I)$ conditional Fourier decay with
\[
	\epsilon(s_x,s_I) = (C\Delta)^{s}\left(\frac{s}{n}\right)^{\frac{t-2}{2}s}\left(\frac{s}{s_x}\right)^{\frac{s_x}{2}}
\]
\end{proposition}
\begin{proof}
	We first prove that $\mu_*$ satisfies the first property from \pref{def:plantedproperties}. Fix $\alpha \subseteq [n]$ with $\abs{\alpha} = s_x \leq d_x$ and $s_I \leq d_I$. Let $A(\alpha) = \{(\beta,\gamma) \mid \widehat{\mu_*\rvert_U}(\alpha,\beta,\gamma) \neq 0\}$ be the set of $\beta,\gamma$ corresponding to the non-zero Fourier coefficients of $\widehat{\mu_*\rvert_U}$. Recall that for the planted density on CSPs the Fourier basis functions for the instance variables are $\zeta_\sigma = \phi_\beta\chi_\gamma$. By the bounds in \pref{lem:conditionalfourier} and the fact that $\E_{\D}[\phi(y_S)] \leq \sqrt{p}^{1-\delta}$, we have
	\begin{align*}
		\sum_{s_I \leq \abs{\sigma} \leq d_I} \Abs{\widehat{\mu_*\rvert_U}(\alpha,\sigma)\Ex[\D]{\zeta_\sigma}}
			&= \sum_{s_I \leq \abs{\beta}, \abs{\bar{\gamma}} \leq d_I} \Abs{\widehat{\mu_*\rvert_U}(\alpha,\beta,\gamma)\Ex[\D]{\phi_\beta \chi_\gamma}}\\
			&\leq \sum_{\substack{s_I \leq \abs{\beta}, \abs{\bar{\gamma}} \leq d_I\\
				(\beta,\gamma) \in A(\alpha)}} 
				\sqrt{pq}^{\abs{\bar{\gamma}\cap\beta}}p^{\bar{\gamma}\setminus\beta}\Abs{\Ex[\D]{\phi_\beta \chi_\gamma}}\\
			&\leq \sum_{\substack{s_I \leq \abs{\beta}, \abs{\bar{\gamma}} \leq d_I\\
				(\beta,\gamma) \in A(\alpha)}} 
				\sqrt{pq}^{\abs{\bar{\gamma}\cap\beta}}p^{\bar{\gamma}\setminus\beta}\sqrt{p}^{(1-\delta)\abs{\beta}}\\
			&\leq \sum_{\substack{s_I \leq \abs{\beta}, \abs{\bar{\gamma}} \leq d_I\\
				(\beta,\gamma) \in A(\alpha)}} 
				p^{(1-\delta)\abs{\bar{\gamma}}}\\
	\end{align*}
	Since $\delta \leq O\left(\frac{1}{\log n}\right)$ and $p = \frac{\Delta n}{\abs{\Enk}} \leq O\left(\frac{1}{n}\right)$ we have
	\[
		p^{1-\delta} = \frac{\Delta n}{\Enk}\cdot p^{-\delta} = \frac{(c \Delta) n}{\Enk}
	\]
	for some constant $c$.
	Further, observe that \pref{lem:conditionalfourier} implies that the non-zero Fourier coefficients of $\mu_*\rvert_{U}$ are a subset of those of $\mu_*$. Thus the number of non-zero Fourier coefficients of $\mu_*\rvert_{U}$ corresponding to a fixed $\alpha$ with $\abs{\bar{\gamma}} \leq r$ is at most $N_r(\alpha)$, and all non-zero Fourier coefficients have $\abs{\bar{\gamma}} \geq \ceil{\frac{\abs{\alpha}}{k}}$. Thus, for $s = \max\left\{\ceil{\frac{s_x}{k}} , s_I\right\}$ we conclude that
	\begin{align*}
		\sum_{s_I \leq \abs{\beta}, \abs{\bar{\gamma}} \leq d_I} \Abs{\widehat{\mu_*\rvert_U}(\alpha,\beta,\gamma)\Ex[\D]{\phi_\beta \chi_\gamma}}
			&\leq \sum_{r = s}^{d_I} N_r(\alpha)p^{(1-\delta)r}\\
			&\leq (C\Delta)^{s}\left(\frac{s}{n}\right)^{\frac{t-2}{2}s + \frac{\abs{\alpha}}{2}}
	\end{align*}
	where the final inequality follows by applying \pref{lem:levell} and absorbing the constant $c$ into $C$.
	Using the fact that $s_x = \abs{\alpha}$ and rearranging yields
	\begin{align*}
		(C\Delta)^{s}\left(\frac{s}{n}\right)^{\frac{t-2}{2}s + \frac{\abs{\alpha}}{2}} 
			&\leq (C\Delta)^{s}\left(\frac{s}{n}\right)^{\frac{t-2}{2}s}\left(\frac{s}{s_x}\right)^{\frac{s_x}{2}} \binom{n}{s_x}^{-\frac{1}{2}}\\
	\end{align*}
	This completes the proof for the first property.
	
	To prove that the second property from \pref{def:plantedproperties} holds, recall that for CSPs the basis functions $\zeta_\sigma = \phi_\beta\psi_\gamma$ where each coordinate $\sigma_j$ indicates both whether a given scope is included in $\beta$, as well as which of the negations on that scope are included in $\gamma$. For any $\sigma$ and $\sigma'$ with disjoint support, let $\beta,\gamma$ and $\beta',\gamma'$ be the corresponding (disjoint) subsets of scopes and negations. Then $\sigma + \sigma'$ corresponds to $\beta\cup\beta',\gamma\cup\gamma'$. Observing that $\norm{\phi_\beta}_\infty = \sqrt{\frac{q}{p}}^{\abs{\beta}}$ and $\norm{\psi_\gamma}_\infty = 1$ we have
	\begin{align*}
		\Abs{\widehat{\mu_*}(\alpha,\sigma + \sigma')}\norm{\zeta_{\sigma'}}_\infty 
			&=\Abs{\widehat{\mu_*}(\alpha,\beta \cup \beta',\gamma \cup \gamma')}\norm{\phi_{\beta'}}_\infty \norm{\psi_{\gamma'}}_\infty \\
			&\leq \sqrt{pq}^{\abs{(\bar{\gamma} \cup \bar{\gamma'})\cap(\beta \cup \beta')}}p^{(\bar{\gamma}\cup\bar{\gamma'})\setminus(\beta\cup\beta')} \cdot \sqrt{\frac{q}{p}}^{\abs{\beta'}}\\
			&\leq \sqrt{pq}^{\abs{\bar{\gamma}\cap\beta}}p^{(\bar{\gamma}\cup\bar{\gamma'})\setminus(\beta\cup\beta')} \leq \sqrt{pq}^{\abs{\bar{\gamma}\cap\beta}}p^{\bar{\gamma}\setminus \beta}.
	\end{align*}
	This is precisely the bound on $\Abs{\widehat{\mu_*}(\alpha,\beta,\gamma)}$ from \pref{lem:fourierformula}, so $\mu_*$ satisfies the second property from \pref{def:plantedproperties}.
\end{proof}

We conclude this section by showing that the function $\epsilon(s_x,s_I)$ is rapidly decaying.
\begin{lemma}
\label{lem:csprapiddecay}
	Fix constants $\nu, \rho > 0$ and let $d_x \leq \rho (d_y - 2b)$ and $d_y \leq \left(n\Delta^{-\frac{2}{t-2}}\right)^{\frac{\nu}{\nu + \rho}}$. Fix $\delta = O(\log n)$ and $b	 \leq \rho d_y$ in \pref{def:plantedproperties}. 
Let $1 < \Delta < n^{\frac{t-2}{2} - \epsilon} $ and let $p = \frac{\Delta n}{\abs{\Enk}}$.
The function $\epsilon(s_x,s_I)$ given in \pref{prop:cspconditionalfourierdecay} is $b$-rapidly decaying. 
\end{lemma}
\begin{proof}
	Observe that by the definition of $\epsilon(s_x,s_I)$ and our choice of $\Delta$ we have $\epsilon(s_x,s_I) = n^{-\epsilon' s_x}$ for some constant $\epsilon' > 0$.
	This implies that $\epsilon(s_x,s_I)$ satisfies the first two properties from \pref{def:rapidlydecaying}.	
	
	For the third property, note that if $d_y - 2b < u \leq d_y$ and $s_x < d_x \leq \rho d_y$ we have
	\begin{align*}
		2^b\epsilon(s_x,u) &= 2^b (C\Delta)^{u}\left(\frac{u}{n}\right)^{\frac{t-2}{2}u}\left(\frac{u}{s_x}\right)^{\frac{s_x}{2}}\\
			&\leq 2^{\rho u} (C\Delta)^{u}\left(\frac{u}{n}\right)^{\frac{t-2}{2}u}\left(\frac{u}{s_x}\right)^{\frac{\rho u}{2}}\\
			&\leq \left(\frac{(C\Delta)^{\frac{2}{t-2}}u^{1 + \rho}}{n}\right)^{u}\\
			&\leq \left(\left(\frac{(C\Delta)^{\frac{2}{t-2}}u}{n}\right)^{1-\nu} \cdot \frac{(C\Delta)^{\nu\frac{2}{t-2}}u^{\nu + \rho}}{n^\nu}\right)^{u}\\
			&\leq \left(\frac{(C\Delta)^{\frac{2}{t-2}}u}{n}\right)^{u(1-\nu)}\\
			&\leq \epsilon(s_x,1)^{1 - \nu}
	\end{align*}
	where the penultimate inequality follows from the fact that $u \leq d_y \leq \left(n\Delta^{-\frac{2}{t-2}}\right)^{\frac{\nu}{\nu + \rho}}$.
\end{proof}

\section{Conjunctive Blockwise-Dense Decompositions}
\label{sec:cbddecomposition}
In this section we show that any distribution can be decomposed into a convex combination of distributions, each of which is conjunctive blockwise dense, along with an error set which is small in an appropriate sense. The proof of this fact is inspired by that in \cite{KothariMR17}.
The main difference is that the error set in their decomposition has small measure under the distribution which was decomposed, whereas in our case the error set only has small measure under the background distribution $\D(p)$. This means our error set may actually contain all the probability mass of the decomposed measure.
However, this allows us to decompose \emph{any} distribution in this way, and is actually necessary for our application.
\begin{lemma}
	\label{lem:cbddecomp}
Let $\D$ be a probability distribution supported on instances $I = (y,b) \in \Enk \times (\Enk \times \{-1,1\}^k)$. Then there is a partition of $\Enk \times (\Enk \times \{-1,1\}^k)$ into subsets $A_1, \cdots A_l, B, C$ such that
\begin{enumerate}
	\item For each $i$ the distribution $\D\rvert_{A_i}$ is $\frac{2}{\delta}t$-CBD with parameter $\delta$.
	\item $\Pr_{\D(p)}\Brac{B} \leq n^{k+1}\left(\frac{p}{2^k}\right)^t$.
	\item $\Pr_{\D}\Brac{C} \leq O(\exp(-n))$.
\end{enumerate}
\end{lemma}
\begin{proof}
	The proof follows from the analysis of a greedy algorithm which constructs the desired partition along with the two error sets.
	\begin{algorithm}
		Let $\D$ be as in the statement of the lemma. We define a recursive function which takes as input a set of instances $A$ and builds up a partition as described in the lemma.
		Initially let $B = C = \emptyset$. 
		%let $C = \{I \mid \abs{I} = (1 \pm \eps)\Delta n\}$ and let $A$ be the set of all instances excluding $C$.
		\paragraph{Decompose($A$)}
		\begin{enumerate}
			\item If $\D\rvert_A$ is blockwise-dense, then add $A$ to the partition and end the recursion.
			\item If $\Pr_{\D(p)}\Brac{A} \leq 2^{-kt}p^t$, then assign $B \gets A\cup B$, and end the recursion.
			\item Set $(A',B') \gets $\textbf{Truncate($A$)}, then assign $B \gets B' \cup B$. If $\Pr_{\D}(A') \leq \exp(-n)$ set $C \gets A' \cup C$ and end the recursion.
			\item Choose $U \subseteq \Enk$ to be a maximal set such that there exist $I^*_U$ with
			\[
				\Pr_{\D\rvert_{A'}}\Brac{I_U = I^*_U} > \Pr_{\D(p)}\Brac{I_U = I^*_U}^{1-\delta}
			\]
			\item Assign $A_0 \gets A' \cap \{I \mid I_U = I^*_U\}$.
			\item Assign $A_1 \gets A' \cap \{I \mid I_U \neq I^*_U\}$.
			\item Add $A_0$ to the partition and call \textbf{Decompose($A_1$)}.
		\end{enumerate}
	\end{algorithm}

	The algorithm calls the following truncation subroutine, which iteratively removes instances from a set $A$ so as to truncate those which appear with much higher probability under $\D$ than under $\D(p)$.
	\begin{algorithm}
		This subroutine takes as input a set of instances $A$ and constructs a partition of $A$ into two sets $A',B'$.
		\paragraph{Truncate($A$)}
		\begin{enumerate}
			\item If $\Pr_{\D}\Brac{A} \leq \exp(-n)$ set $A' \gets A$ and terminate.
			\item If for all $I' \in A$ we have $\Pr_{\D\rvert_A}\Brac{I = I'} < 2^{kt}p^{-t}\Pr_{\D(p)\rvert_A}\Brac{I = I'}$, set $A' \gets A$ and terminate.
			\item Else choose $I^*$ which maximizes $\frac{\Pr_{\D\rvert_A}\Brac{I = I^*}}{\Pr_{\D(p)\rvert_A}\Brac{I = I^*}}$ and set $B' \gets B' \cup \{I^*\}$.
			\item Call \textbf{Truncate($A\setminus I^*$)}.
		\end{enumerate}
	\end{algorithm}
	We prove the lemma through a series of claims.
	\begin{claim}In any execution of \textbf{Decompose} that makes it to the end of the function we must have $\D\rvert_{A_0}$ is conjunctive blockwise-dense with fixed block $U$.
	\end{claim}
	\begin{proof}
		To see why let $V \subseteq \Enk\setminus U$ and suppose that there exists $I'_V$ such that $\Pr_{\D\rvert_{A_0}}\Brac{I_V = I'_V} > \Pr_{\D(p)}\Brac{I_V = I'_V}^{1-\delta}$. Then
		\begin{align*}
			\Pr_{\D\rvert_A}\Brac{I_{U \cup V} = (I^*_U,I'_U)}
				&\geq \Pr_{\D(p)}\Brac{I_U = I^*_U}^{1-\delta} \cdot \Pr_{\D(p)}\Brac{I_V = I'_V}^{1-\delta}\\
				&= \Pr_{\D(p)}\Brac{I_{U \cup V} = (I^*_U, I'_U)}^{1-\delta}
		\end{align*}
	which contradicts the maximality of $U$.
	\end{proof}
	Next we show that \textbf{Truncate} does not remove too much probability mass from $\D(p)$.
	\begin{claim}
		\label{claim:Bsmall}
		If $(A',B') =$ \textbf{Truncate($A$)}, then $\Pr_{\D(p)}\Brac{A'} \geq \left(1 - n\left(\frac{p}{2^k}\right)^t\right)\Pr_{\D(p)}\Brac{A}$.
	\end{claim}
	\begin{proof}
%		Since $A$ is supported only on instances $I$ with $\abs{I} = (1 \pm \epsilon)\Delta n$ we have that, for any $I'\in A$,
%		\[
%			\Pr_{\D(p)}\brac{I = I'} \leq \left(\frac{p}{2^k}\right)^{(1-\eps)\Delta n}.
%		\]
		Let $A_i$ denote the input to the $i$-th recursive call to \textbf{Truncate}.
		Now note that in each recursive call if $\Pr_{\D(p)}\Brac{A_i}$ decreases by a factor of $(1 - \eta_i)$ after removing $I^*$,
		then $\Pr_{\D}\Brac{A_i}$ decreases by a factor of $(1 - \eta_i\left(\frac{p}{2^k}\right)^{-t})$. Thus if after the $j$-th recursive call we have $\sum_{i = 1}^j \eta_i = n\left(\frac{p}{2^k}\right)^t$ then,
		\[
			\Pr_{\D(p)}\Brac{A_i} = \exp\left(-n\left(\frac{p}{2^k}\right)^t\right)\Pr_{\D(p)}\Brac{A_0} \geq \left(1 - n\left(\frac{p}{2^k}\right)^t\right)\Pr_{\D(p)}\Brac{A_0}
		\]
		while on the other hand
		\[
			\Pr_{\D}\Brac{A_i} = \exp(-n)\Pr_{\D}\Brac{A_0} \leq \exp(-n).
		\]
		Since the subroutine terminates once $\Pr_{\D}\Brac{A_i} \leq \exp(-n)$, the claim is proven.
	\end{proof}
	Next we show that the fixed block $U$ is not too large.
	\begin{claim} $\abs{U} \leq \frac{2}{\delta}t$.\end{claim}
	\begin{proof}
	If the call to \textbf{Truncate} does not return a set $A'$ with $\Pr_{\D}\Brac{A'} \leq \exp(-n)$, then we must have for all $I'$
	\[
		\Pr_{\D\rvert_{A'}}\Brac{I = I'} < 2^{kt}p^{-t}\Pr_{\D(p\rvert_{A'})}\Brac{I = I'}.
	\]
	Thus, for any restricted instance $I'_U$ after the truncation step
	\begin{align*}
		\Pr_{\D\rvert_{A'}}\Brac{I_U = I'_U}
			&\leq 2^{kt}p^{-t}\Pr_{\D(p)\rvert_{A'}}\Brac{I_U = I'_U}.
	\end{align*}
	Thus, we have that for the instance $I^*_U$,
	\[
		\Pr_{\D(p)}\Brac{I_U = I^*_U}^{1-\delta}\leq 2^{kt}p^{-t}\Pr_{\D(p)\rvert_{A'}}\Brac{I_U = I^*_U} = 2^{kt}p^{-t}\frac{\Pr_{\D(p)\rvert_{A'}}\Brac{I_U = I^*_U}}{\Pr_{\D(p)}\Brac{A'}}.
	\]
	Since we have passed the second step in the algorithm, $\Pr_{\D(p)}\Brac{A'} \geq 2^{-kt}p^t$. Using this fact and rearranging terms gives
	\[
		2^{2kt}p^{-2t} \geq \Pr_{\D(p)}\Brac{I_U = I^*_U}^{-\delta} \geq 2^{\delta k\abs{U}}p^{-\delta\abs{U}},
	\]
	where the last inequality comes from the definition of $\D(p)$. Thus $\abs{U} \leq \frac{2}{\delta}t$.
	\end{proof}
	To wrap up the proof of the lemma, note that \textbf{Decompose} is called recursively at most $n^k$ times, because we fix at least one constraint scope in every call that does not terminate. Thus, by \prettyref{claim:Bsmall} the total probability mass added to $B$ is at most $n^{k+1}\left(\frac{p}{2^k}\right)^t$.
\end{proof}

\section{Proof of LP Lower Bounds}
\label{sec:lbproof}
In this section we prove lower bounds for LP formulations of CSPs. We proceed with the approach based on pseudo-calibration as introduced in \prettyref{sec:proofoverview}. In particular, assume there exists an LP formulation for a CSP which certifies the upper bound $c$ on a subset of instances $A \subset \{-1,1\}^{\abs{\Enk}} \times \{-1,1\}^{\abs{\Enk}k}$ of measure $s$ under $\D(p)$. We will identify a subset $B \subseteq A$ of instances and some $\lambda > 0$ such that both
\begin{equation}
\label{eqn:lhs}
	\Ex[\{-1,1\}^n\times\D(p)]{\Ind((y,b) \in B)\bar{\mu}_*(x,y,b)(c - F(x,y,b))} < -\lambda
\end{equation}
and
\begin{equation}
\label{eqn:rhs}
	\Ex[\{-1,1\}^n\times\D(p)]{\Ind((y,b) \in B)\bar{\mu}_*(x,y,b) \sum_{i=0}^{R} p_i(y,b)q_i(x)} \geq -\lambda.
\end{equation}
This will then contradict the fact that the LP formulation certifies the upper bound $c$ on the instances in $A$.

\paragraph{Setting Parameters}For the rest of this section let $\nu_x > \nu_y > 0$ be small constants which will be chosen later to be sufficiently small. Let $d = (d_x,d_y)$ where
\[
	d_x = \left(\frac{Cn^{t-2}}{\Delta^2}\right)^{\frac{(1 - \nu_x)}{k}} \qquad d_y = \left(\frac{Cn^{t-2}}{\Delta^2}\right)^{\frac{(1 - \nu_y)}{k}}
\]
Let $\bar{\mu}_*$ be the $d$-pseudo-calibration of $\mu_*$.
Additionally, set $p = \frac{\Delta n}{\abs{\Enk}}$ where $1 < \Delta < n^{\frac{t-2}{2} - \eps}$.

The first step will be to prove that $\Ex[x]{\bar{\mu}_*(x,y,b)(c - F(x,y,b))}$ is negative with high probability over instances $I = (y,b)$.
\begin{lemma}
	\label{lem:lhsnegative}
	Let $c = (1 - \eta)\Delta n$ and let $m = m(y,b)$ be the number of constraints in the instance $(y,b)$. 
	Let $E$ be the event that both:
	\begin{enumerate}
		\item $\abs{m(y,b) - \Delta n} \leq \frac{\eta}{2}\Delta n$
		\item $\Ex[x]{\bar{\mu}_*(x,y,b)(c - F(x,y,b))} \leq - \frac{\eta}{2}\Delta n$.
	\end{enumerate}
	Then
	\[
		\Pr_{(y,b) \sim \D(p)}\Brac{E} = 1 - o(1).
	\]
\end{lemma}
\begin{proof}
	The objective function $F$ has degree at most $k$ in $x$ and $b$ and degree one in $y$. Thus
	\[
		\Ex[\{-1,1\}^n\times\D(p)]{\bar{\mu}_*(x,y,b)(c - F(x,y,b))}
			= \Ex[\D_*]{c - F(x,y,b)} = -\eta\Delta n.
	\]
	Next let $G(y,b) = \Ex[x]{\bar{\mu}_*(x,y,b)F(x,y,b))}$. By the previous equation $\Ex[y,b]{G(y,b)} = \Delta n$. We next prove that $G$ concentrates around its mean.

	For a scope $S$ let $\mu_S(y,b) = \Ex[x]{P(b_S \circ x_S)\bar{\mu}_*(x,y,b)}$.
	Therefore
	\[
		G(y,b) = \sum_S \Ind(y_S = -1)\mu_S(y,b).
	\]
	Next we compute the Fourier expansion of $\mu_S$. For $T \subset [k]$ we will use the notation $\alpha_T = \{S_j \mid j \in T\}$.
	\begin{align*}
		\mu_S(x,y,b)
			&= \Ex[x]{\left(\sum_{T \subseteq [k]}\widehat{P}(T)\prod_{j\in T}x_{S_j} b_{(S,j)}\right)\bar{\mu}_*}\\
			&= \sum_{T \subseteq [k]}\widehat{P}(T)\prod_{j\in T}b_{(S,j)} \sum_{\substack{\gamma \vdash \alpha_T\\ \beta \subseteq \bar{\gamma}}} \widehat{\mu}(\alpha_T,\beta,\gamma)\phi_\beta(y)\psi_\gamma(b)\\
			&= \sum_{T \subseteq [k]}\widehat{P}(T) \sum_{\substack{\gamma \vdash \alpha_T\\ \beta \subseteq \bar{\gamma}}} \widehat{\mu}(\alpha_T,\beta,\gamma)\phi_\beta(y)\psi_{\gamma\triangle \{S\} \times T}(b)
	\end{align*}
	where the second sum runs over $\gamma$ with $\abs{\bar{\gamma}}\leq d_y$.
	For each fixed $T$ in the above expression, the only terms with $\abs{\bar{\gamma}} \leq 2$ occur for pairs $\beta,\gamma$ where $\gamma = \{S\}\times T$ and $\beta = \{S\}$ or $\beta = \emptyset$.
	Further note that all the basis functions $\phi_\beta(y)\psi_{\gamma\triangle \{S\} \times T}(b)$ are distinct except for the cases when $\gamma = \{S\}\times T$ and $\beta = \emptyset$.
	These cases correspond exactly to the constant term in the Fourier expansion of $\mu_S$ since $\gamma \triangle \{S\}\times T = \emptyset$ when $\gamma = \{S\}\times T$.
	
	Finally, if we fix $y_S = -1$, all the terms with $\gamma = \{S\}\times T$ and $\beta = \{S\}$ will be fixed to constants.
	Other basis functions may add in pairs, but this will at most double the magnitude of the Fourier coefficients by the second property from \pref{def:plantedproperties}.
	Thus, after fixing $y_S = -1$, the non-constant part of $\mu_S$ has Fourier coefficients bounded in magnitude by $2\Abs{\widehat{P}(T)\widehat{\mu}(\alpha_T,\beta,\gamma)}$. As for the constant term, note that both
	\[
		\Ex[y,b]{\mu_S \Ind(y_S = -1)} = \Ex[y,b]{\mu_S \mid y_S = -1} p
	\]
	and
	\[
		\Ex[y,b]{\mu_S \Ind(y_S = -1)} = \Ex[\D_*]{\Ind(y_S = -1)P(b_S \circ x_S)} = p.
	\]
	Dividing the first equation by the second then implies that the constant term after fixing $y_S = -1$ is $\Ex[y,b]{\mu_S \mid y_S = -1} = 1$.
	
	Since $P$ is $\{0,1\}$-valued, $\Abs{\widehat{P}(T)} \leq 1$ for all $T$. There are at most $2^k$ subsets $T \subseteq [k]$, so by \pref{lem:l2norm}
	\[
		\Ex[y,b]{\mu_S^2 \mid y_S = -1} \leq 1 + O\left(\Delta n^{-\frac{t - 2}{2}}\right) \leq 1 + O(n^{-\eps}).
	\]

	By a similar argument, we have that if we fix a pair $y_S,y_S'$ the non-constant part of $\mu_S$ has Fourier coefficients bounded by $4\Abs{\widehat{P}(T)\widehat{\mu}(\alpha_T,\beta,\gamma)}$ and the constant term is again equal to one. Therefore
	\[
		\Ex[y,b]{\mu_S \mu_{S'} \mid y_S = y_{S'} = -1}
			\leq \sqrt{\Ex{\mu_S^2 \mid y_S = y_{S'} = -1}\Ex{\mu_{S'}^2 \mid y_S = y_{S'} = -1}} \leq 1 + O(n^{-\eps}).
	\]
	Thus we conclude that
	\begin{align*}
		\Ex[y,b]{G^2}
			&= \sum_{S,S'}\Ex{\Ind(y_S = y_S' = -1)\mu_S \mu_{S'} \mid y_S = y_{S'} = -1}p^2\\
			&\leq p\Abs{\Enk}(1 + O(n^{-\epsilon})) + p^2\Abs{\Enk}^2(1 + O(n^{-\eps})) \leq (\Delta n)^2 (1 + n^{-\eps'})
	\end{align*}
	for some constant $\eps' > 0$.
	
	Chebyshev's inequality then implies that $\Pr[\Abs{G - \Delta n} \geq n^{-\frac{\eps'}{4}}\Delta n] \leq n^{-\frac{\eps'}{2}}$. For any constant $\eta$ this implies that $\Ex[x]{\bar{\mu}_*(x,y,b)(c - F(x,y,b))} \leq - \frac{\eta}{2}\Delta n$.

	Lastly, by \prettyref{lem:constraintconcentration} $\abs{m(y,b) - \Delta n} \leq \frac{\eta}{2}\Delta n$ with probability at least $1 - 2\exp\left(\Omega(\eta^2\Delta n)\right) = 1 - o(1)$.
\end{proof}
%The next lemma will rely on the following concentration bound based on hypercontractivity.
%\begin{theorem}[\cite{OD14} Chapter 9]
%	\label{thm:lowdegconcentration}
%	Let $f:\{-1,1\}^n \to \R$ have degree at most $k$. Then for any $t \geq \sqrt{2e}^k$ we have
%	\[
%		\Pr_{x \sim \{-1,1\}^n}\Brac{\abs{f(x)} \geq t \norm{f}_2} \leq \exp\left(-\frac{k}{2e}t^{\frac{2}{k}}\right).
%	\]
%\end{theorem}
In the next lemma we plug in our parameter settings to compute a bound on the probability that $\bar{\mu}_*$ is non-negative after averaging over a CBD distribution.
\begin{lemma}
	\label{lem:nonnegwhp}
	For $\rho > 0$ let $b \leq \min(\rho d_y , k^{-1} d_x)$. Let $\D$ be a $b$-CBD distribution with parameter $\delta \leq O\left(\frac{1}{\log n}\right)$ and fixed block $U$. Let $H(x) = \E_{\D}\Brac{\bar{\mu}_*}$. Then there is positive number $\nu = \nu(\rho)$ going to $0$ as $\rho \to 0$ such that
	\[
		\Pr_{x}\Brac{H(x) \leq -o(1)} \leq O\left(\exp\left(-\left(\frac{Cn^{t-2}}{\Delta^2}\right)^{\frac{(1 - \nu)}{k}}\right)\right).
	\]
\end{lemma}
\begin{proof}
%	We begin by breaking up the probability by conditioning on each possible setting of the variables $x_U$.
%	\begin{equation}
%		\label{eqn:sumrestrictions}
%		\Pr_{x}\Brac{H(x) \leq -O\left(n^{-\frac{\eps}{4k}}\right)} = \frac{1}{2^{k\abs{U}}}\sum_{x'_{U}}\Pr_{x_{U^c}}\Brac{H(x'_U,x_{U^c}) \leq -O\left(n^{-\frac{\eps}{4k}}\right)}
%	\end{equation}
%	Next let $H_{x'_U}(x) = H(x'_U,x)$. Since $\D$ is supported only on instances where $\bar{\mu}_*(x,y,b)$ is a Sherali-Adams pseudo-density we have that
%	\begin{align*}
%		\Ex[x]{H_{x'_U}(x)} &= \Ex[x]{H(x)\Ind(x_U = x'_U)}\\
%			&= \Ex[\D]{\Ex[x]{\bar{\mu}_*(x,y,b)\Ind(x_U = x'_U)}}\\
%			&\geq 0
%	\end{align*}
%	because the inner expectation is the average of a Sherali-Adams pseudo-density times a non-negative junta on $kd \leq d_x$ variables.
	By \pref{prop:cspconditionalfourierdecay} the density $\mu_*$ exhibits $\epsilon(s_x,s_I)$-Fourier decay, and \pref{lem:csprapiddecay} the function $\epsilon(s_x,s_I)$ is rapidly decaying. Thus we may apply \pref{prop:nonnegwhp} to conclude that
	\begin{align*}
		\Pr_{x}\Brac{H(x) \leq -o(1)} &\leq \sum_{s = 1}^{d_x}\exp\left(-\frac{s}{2e}\epsilon(s,1)^{-\frac{2 - 2\nu}{s}}\right)\\
			&= \sum_{s = 1}^{d_x}\exp\left(-\frac{s}{2e}\epsilon\left(s,\ceil*{\frac{s}{k}}\right)^{-\frac{2 - 2\nu}{s}}\right)\\
			&\leq \sum_{s = 1}^{d_x}\exp\left(-\frac{s}{2e}\left(\frac{(C\Delta)^{\frac{2}{t-2}}s}{kn}\right)^{-\frac{2 - 2\nu}{s} \cdot \frac{(t-2)s}{k}}\right)\\
			&\leq \sum_s \exp\left(-\left(\frac{Cn^{t-2}}{\Delta^2}\right)^{\frac{(1 - \nu)}{k}}s^{1 - \frac{1}{k}}\right)\\
			&\leq O\left(-\exp\left(\left(\frac{Cn^{t-2}}{\Delta^2}\right)^{\frac{(1 - \nu)}{k}}\right)\right)
	\end{align*}

\end{proof}
We are now ready to prove our main result.
\begin{theorem}
	%\label{thm:lplowerbound}
Let $P$ be a predicate supporting a $(t-1)$-wise uniform distribution on satisfying assignments.
For any $\eps > 0$, let $1 < \Delta < n^{\frac{t-2}{2} - \eps}$.
Then for any constants $\eta, \nu > 0$, any linear programming formulation that $\eta$-refutes random instances of $\CSP(P)$ with constant probability must have size at least $\Omega\left(\exp\left(\left(\frac{Cn^{t-2}}{\Delta^2}\right)^{\frac{(1 - \nu)}{k}}\right)\right)$.
\end{theorem}
\begin{proof}
	 For any $\nu > 0$, suppose that there is a linear programming formulation of size $R \leq O\left(\exp\left(\left(\frac{Cn^{t-2}}{\Delta^2}\right)^{\frac{(1 - \nu)}{k}}\right)\right)$ for $\eta$-refuting random instances of $\CSP(P)$. By \prettyref{lem:nonnegativerep} this means there exists a set $A$ of constant measure under $D(p)$, and non-negative functions $p_i(y,b),q_i(x)$ such that, for $c = (1-\eta)\Delta n$:
	\begin{equation}
		\label{eqn:maineqn}
		c - F(x,y,b) = \sum_{i=0}^R p_i(y,b)q_i(x)
	\end{equation}
	whenever $(y,b) \in A$. We will derive a contradiction by multiplying each side of the above equation by $\bar{\mu}_*$ and then averaging over $(x,y,b)$.

	Let $m(y,b)$ be the number of constraints in $(y,b)$ and let $E$ be the event defined in \prettyref{lem:lhsnegative}. By the lemma, $\Pr_{\D(p)}\Brac{E} = 1 -o(1)$. Thus, letting $A' = A \cap E$ we have $\Pr_{\D(p)}\Brac{A'}$ is constant, because $A$ has constant measure. We now restrict \prettyref{eqn:maineqn} to instances in $A'$ and normalize the equation by dividing both sides by $\Delta n$. Note that by the definition of $E$ we have $\abs{\frac{1}{\Delta n}(c - F(x,y,b))} \leq 1 + \frac{\eta}{2}$ for all $x$, and for all $(y,b)\in E$.

	Thus, we must have the same bound for the right hand side of \prettyref{eqn:maineqn}
	\begin{equation}
		\label{eqn:normalize}
		\frac{1}{\Delta n}\sum_{i=0}^R p_i(y,b)q_i(x) \leq 1 + \frac{\eta}{2}
	\end{equation}
	Now, since the $p_i,$ are non-negative, we can re-normalize each $p_i$ to be a density relative to $\D(p)$ by simply rescaling $p_i$ by $\Ex{p_i}^{-1}$ and $q_i$ by $\Ex{p_i}$. After this rescaling, averaging \prettyref{eqn:normalize} over $(y,b)$ implies that $\sum_i q_i(x) \leq 1 + \frac{\eta}{2}$ for all $x$.

	Let $\D_i$ be the probability distribution given by the density $p_i$, let $\delta \leq O\left(\frac{1}{\log n}\right)$ and let $r = \frac{d_x}{2k\log n}$.
	By \prettyref{lem:cbddecomp} we can partition $A'$ into sets $A_1,\dots A_N,B_i,C_i$ such that each $A_j$ is a $\frac{2}{\delta}r$-CBD distribution with parameter $\delta$, $\Pr_{\D(p)}\Brac{B_i} \leq n^{k+1}\left(\frac{p}{2^k}\right)^r$, and $\Pr_{\D_i}\Brac{C_i} \leq O(\exp(-n))$. 
	Let $\rho > 0$ be arbitrarily small. Recalling the parameter settings for $d_x$ and $d_y$ from the beginning of this section, we can set $\nu_y < \nu_x$ to be small enough that $\rho d_y > k^{-1} d_x > \frac{2}{\delta}r$.
	In particular, we have that each density is $d$-CBD for $d \leq \min(\rho d_y,k^{-1}d_x)$. Letting $B'_i \defeq A' \cap B_i$ and $H_{i,j}(x) \defeq \Ex[\D_i\rvert A_j]{\bar{\mu}_*}$  yields
	\begin{align*}
		\Ex{\Ind((y,b)\notin B'_i)\bar{\mu}_*(x,y,b)p_i(y,b) } &= \sum_j \Pr_{\D_i}[A_j] \Ex[\D_i\rvert A_j]{\bar{\mu}_*} + \Pr_{\D_i}[C_i]\Ex[\D_i\rvert C_i]{\bar{\mu}_*}\\
		&= \sum_j \Pr_{\D_i}[A_j] H_{i,j}(x) + \Pr_{\D_i}[C_i]\Ex[\D_i\rvert C_i]{\bar{\mu}_*}
	\end{align*}
	Since $\Pr_{\D_i}\Brac{C_i} \leq O(\exp(-n))$ the second term above is bounded in magnitude by
	\begin{equation}
		\label{eqn:Cerror}
		\abs{\Pr_{\D_i}[C_i]\Ex[\D_i\rvert C_i]{\bar{\mu}_*}} \leq O(\exp(-n)) \norm{\bar{\mu}_*}_\infty \leq O(\exp(-n)) \cdot n^{O(d_x + d_y)} \leq \exp(-\Omega(n))
	\end{equation}
	where the final inequality uses the fact that both $d_x$ and $d_y$ are bounded by $O\left(n^{(1 - \nu_y)}\right)$.
	Further, since each distribution $\D_i$ is $d$-CBD for $d \leq \min(\rho d_y,k^{-1}d_x)$ we have by \prettyref{lem:nonnegwhp} that there exists $\nu' = \nu'(\rho)$ going to $0$ as $\rho \to 0$ such that for each $i,j$
	\[
		\Pr_{x}\Brac{H_{i,j}(x) \leq -o(1)} \leq O\left(\exp\left(-\left(\frac{Cn^{t-2}}{\Delta^2}\right)^{\frac{(1 - \nu')}{k}}\right)\right) \leq o(R^{-1}n^{-d_x}).
	\]
	where the last inequality follows from the definition of $d_x$, and by taking $\nu'$ to be sufficiently small by choosing $\rho$ to be sufficiently small.
	Since the $\frac{1}{\Delta n}\sum_i q_i(x) \leq 1 + \frac{\eta}{2}$ we conclude that
	\begin{align*}
		\sum_{i,j}\Pr_{\D_i}\Brac{A_j}\Ex[x]{q_i(x)H_{i,j}(x)}
			&\geq \sum_{i,j}\Pr_{\D_i}\Brac{A_j} \cdot \left(-o(1)\Ex{q_i(x)} - \Pr_{x}\Brac{H_{i,j}(x) \leq -o(1)}\norm{H_{i,j}}_\infty\right)\\
			&\geq -o(1) - R\max_{i,j}\left(\Pr_{x}\Brac{H_{i,j}(x) \leq -o(1)}\norm{H_{i,j}}_\infty\right)\\
			&\geq -o(1) - R \cdot o(R^{-1}n^{-d_x}) \cdot n^{d_x} \geq -o(1)
	\end{align*}
	Let $B = \cup_i B'_i$. Then the above inequality combined with \prettyref{eqn:Cerror} implies that
	\[
		\sum_i \Ex{\Ind((y,b) \notin B)\bar{\mu}_*(x,y,b)p_i(y,b)q_i(x)} \geq -o(1).
	\]
	To summarize, averaging the right hand side of (the normalized version of) \prettyref{eqn:maineqn} multiplied by $\bar{\mu}_*$ and restricted to instances not in $B$ has expectation that is at least a small negative number. Now observe that by our choice of $r$ and $R$ we have
	\[
		\Pr_{\D(p)}[B] \leq Rn^{k+1}\left(\frac{p}{2k}\right)^r \leq R \exp(-O(d_x)) \leq o(1)
	\]
	where again we require that $\nu_x > 0$ is taken to be sufficiently small relative to $\nu$. 
	Recall on the left hand side of \prettyref{eqn:maineqn} we have that
	\[
		\frac{1}{\Delta n}\Ex[x]{\bar{\mu}_*(x,y,b)(c - F(x,y,b))} \leq - \frac{\eta}{2}
	\]
	whenever $(y,b)\in A'$. Further, $B \subseteq A'$ and $A'$ has constant measure under $\D(p)$, so $B$ cannot be all of $A'$. We conclude that
	\[
		\frac{1}{\Delta n}\Ex[x,y,b]{\Ind((y,b)\notin B)\bar{\mu}_*(x,y,b)(c - F(x,y,b))} \leq - \frac{\eta}{2}
	\]
	For any constant $\eta$ this yields the desired contradiction.
\end{proof}

\appendix

\bibliographystyle{amsalpha}
\bibliography{bib/mr,bib/dblp,bib/scholar,bib/bibliography,bib/lpsize,bib/lplowerbounds,bib/writeup}

%\printbibliography

\end{document}